\documentclass[a4paper,USenglish]{lipics}
 
\usepackage{microtype}
\title{Towards a Universal Approach for Monotonic Searchability in Self-Stabilizing Overlay Networks\footnote{This work was partially supported by the German Research Foundation (DFG) within the Collaborative Research Center ``On-The-Fly Computing'' (SFB 901).}}
\titlerunning{Towards a Universal Approach for Monotonic Searchability}

\author[1]{Christian Scheideler}
\author[2]{Alexander Setzer}
\author[3]{Thim Strothmann}
\affil[1]{Paderborn University\\
  Fürstenallee 11, Paderborn, Germany\\
  \texttt{scheideler@uni-paderborn.de}}
\affil[2]{Paderborn University\\
  Fürstenallee 11, Paderborn, Germany\\
  \texttt{asetzer@mail.uni-paderborn.de}}
\affil[3]{Paderborn University\\
  Fürstenallee 11, Paderborn, Germany\\
  \texttt{thim@mail.uni-paderborn.de}}
\authorrunning{C. Scheideler and A. Setzer and T. Strothmann}

\Copyright{Christian Scheideler, Alexander Setzer, and Thim Strothmann}

\subjclass{C.2.4 Distributed Systems}
\keywords{Monotonic Self-Stabilization, Topological Self-Stabilization, Overlay Networks}

\usepackage{xcolor}
\usepackage{xspace}
\usepackage{algorithm}
\usepackage[noend]{algpseudocode}
\usepackage{xifthen}
\usepackage{tikz}
\usepackage{cancel}
\algrenewcommand{\algorithmiccomment}[1]{$\rhd$ #1}
\algrenewcommand\algorithmicprocedure{\textbf{action}}
\newcommand{\impldelegate}[1]{\textsc{ImplDelegate}(\ensuremath{#1})\xspace}
\newcommand{\timeout}{\textsc{Timeout}\xspace}
\newcommand{\search}[1]{\textsc{Search}(\ensuremath{#1})\xspace}
\newcommand{\initsearch}[1]{\textsc{InitiateNewSearch}(\ensuremath{#1})\xspace}
\newcommand{\forwardprobe}[1]{\textsc{ForwardProbe}(\ensuremath{#1})\xspace}
\newcommand{\probe}[1]{\textsc{Probe}(\ensuremath{#1})\xspace}
\newcommand{\fastprobe}[1]{\textsc{FastProbe}(\ensuremath{#1})\xspace}
\newcommand{\psuccess}[1]{\textsc{ProbeSuccess}(\ensuremath{#1})\xspace}
\newcommand{\pfail}[1]{\textsc{ProbeFail}(\ensuremath{#1})\xspace}

\newcommand{\pisf}{\ensuremath{\mathcal{P_{ISF}}}\xspace}
\newcommand{\pidf}{\ensuremath{\mathcal{P_{IDF}}}\xspace}

\xspace
\xspace 

\newcommand{\templeft}[1][]{\ensuremath{Temp_{L}\ifthenelse{\isempty{#1}}{}{(#1)}}\xspace}
\newcommand{\tempright}[1][]{\ensuremath{Temp_{R}\ifthenelse{\isempty{#1}}{}{(#1)}}\xspace}

\newcommand{\delegateREQ}[1]{\textsc{DelegateREQ}(#1)\xspace}
\newcommand{\delegateACK}[1]{\textsc{DelegateACK}(#1)\xspace}

\newcommand{\mss}{$ms$-sufficient\xspace}
\newcommand{\invar}{\ensuremath{\mathcal{I}}\xspace}
\newcommand{\mdl}{\textsc{mdl}\xspace}

\newcommand{\skipp}{\ensuremath{SKIP^{\mbox{+}}}}

\begin{document}

\maketitle

\begin{abstract}
For overlay networks, the ability to recover from a variety of problems like membership changes or faults is a key element to preserve their functionality.
In recent years, various self-stabilizing overlay networks have been proposed that have the advantage of being able to recover from \emph{any} illegal state.
However, the vast majority of these networks cannot give any guarantees on its functionality while the recovery process is going on.
We are especially interested in \emph{searchability}, i.e., the functionality that search messages for a specific identifier are answered successfully if a node with that identifier exists in the network.
We investigate overlay networks that are not only self-stabilizing but that also ensure that \emph{monotonic} searchability is maintained while the recovery process is going on, as long as there are no corrupted messages in the system. 
More precisely, once a search message from node $u$ to another node $v$ is successfully delivered, all future search messages from $u$ to $v$ succeed as well. 
Monotonic searchability was recently introduced in OPODIS 2015, in which the authors provide a solution for a simple line topology.
 We present the first \emph{universal} approach to maintain monotonic searchability that is applicable to a wide range of topologies.
As the base for our approach, we introduce a set of primitives for manipulating overlay networks that allows us to maintain searchability and show how existing protocols can be transformed to use theses primitives.
 We complement this result with a generic search protocol that together with the use of our primitives guarantees monotonic searchability.
 As an additional feature, searching existing nodes with the generic search protocol is as fast as searching a node with any other fixed routing protocol once the topology has stabilized. 
 \end{abstract}

 \ 

\hfill
\newline
This is the full version of a correspondent paper appearing in the DISC 2016 proceedings.
 
\newpage

\section{Introduction}
The Internet has opened up tremendous opportunities for people to interact and exchange information.
 Particularly popular ways to interact are peer-to-peer systems and social networks. 
 For these systems it is very important that they are highly available. 
 However, once these systems become large enough, changes and faults are not an exception but the rule.
Therefore, mechanisms are needed that ensure that whenever problems occur, they are quickly repaired, and all parts of the system that are still functional should not be affected by the repair process. 
Protocols that are able to recover from arbitrary states are also known as \emph{self-stabilizing} protocols.

Since the seminal paper by Dijkstra in 1974~\cite{Dijkstra74}, self-stabilizing protocols have been investigated for many classical problems including leader election, consensus, matching, clock synchronization and token distribution problems.
Recently, also various protocols for self-stabilizing overlay networks have been proposed (e.g., \cite{DolevT2013,corona,JacobRSS2012,KniesburgesKS12}). 
However, for all of these protocols it is only known that they \emph{eventually} converge to the desired solution, but the convergence process is not necessarily \emph{monotonic}. 
In other words, it is not ensured for two points in time $t, t'$ with $t<t'$ that the functionality of the topology at time $t'$ is better than the functionality at time $t$.

In this paper, we continue our research started in~\cite{SetzerSSS14} and investigate protocols for self-stabilizing overlay networks that guarantee the \emph{monotonic} preservation of a characteristic that we call \emph{searchability}, i.e., once a search message from node $u$ to another node $v$ is successfully delivered, all future search messages from $u$ to $v$ succeed as well. 
Searchability is a useful and natural characteristic for an overlay network since searching for other participants is one of the most common tasks in real-world networks.
Moreover, a protocol that maintains monotonic searchability has the huge advantage that in every state, even if the self-stabilization process has not converged yet, the parts of the topology that are already in a legal state can be used for search requests.

Instead of focusing on a specific topology, as done in~\cite{SetzerSSS14}, we present an approach that is aimed at universality. 
As a base, we present a set of primitives for overlay network maintainance for which we prove that they enable monotonic searchability.
On top of that, we give a generic search protocol that, together with a protocol that solely uses these primitives, guarantees monotonic searchability.
Additionally, we show that existing self-stabilizing overlay network protocols can be transformed to use our primitives.

To the best of our knowledge, we are the first to investigate monotonic searchability as an attempt to explore maintaining properties beyond the traditional "time and space" metrics during stabilization.  We believe that the question of how to maintain monotonic searchability and similar properties during topological stabilization has a lot of potential for future research. 

\subsection{Model}
We consider a distributed system consisting of a fixed set of nodes in which each node has a unique reference and a unique immutable numerical identifier (or short id).
The system is controlled by a protocol that specifies the variables and actions that are available in each node. 
In addition to the protocol-based variables there is a system-based variable for each node called \emph{channel} whose values are sets of messages. 
We denote the channel of a node $u$ as $u.Ch$ and it contains all incoming messages to $u$. 
Its message capacity is unbounded and messages never get lost.
A node can add a message to $u.Ch$ if it has a reference of $u$.
Besides these channels there are no further communication means, so only point-to-point communication is possible.

There are two types of \emph{actions} that a protocol can execute. 
The first type has the form of a standard procedure $\langle label\rangle (\langle parameters \rangle): \langle command \rangle$, where $label$ is the unique name of that action, $parameters$ specifies the parameter list of the action, and $command$ specifies the statements to be executed when calling that action. 
Such actions can be called locally (which causes their immediate execution) and remotely. 
In fact, we assume that every message must be of the form $\langle label \rangle (\langle parameters \rangle)$, where $label$ specifies the action to be called in the receiving node and $parameters$ contains the parameters to be passed to that action call. 
All other messages are ignored by nodes.  
The second type has the form $ \langle label\rangle: \langle guard \rangle \longrightarrow \langle command \rangle$, where $label$ and $command$ are defined as above and $guard$ is a predicate over local variables. 
We call an action whose guard is simply \textbf{true} a \emph{timeout} action.

The \emph{system state} is an assignment of values to every variable of each node and messages to each channel. 
An action in some node $u$ is \emph{enabled} in some system state if its guard evaluates to \textbf{true}, or if there is a message in $u.Ch$ requesting to call it.
In the latter case, when the corresponding action is executed, the message is processed (in which case it is removed from $u.Ch$).
An action is \emph{disabled} otherwise. 
Receiving and processing a message is considered as an atomic step.

A \emph{computation} is an infinite fair sequence of system states such that for each state $S_i$, the next state $S_{i+1}$ is obtained by executing an action that is enabled in $S_i$. 
This disallows the overlap of action execution, i.e., action execution is \emph{atomic}. 
We assume \emph{weakly fair action execution} and \emph{fair message receipt}. 
Weakly fair action execution means that if an action is enabled in all but finitely many states
of a computation, then this action is executed infinitely often. 
Note that a timeout action of a node is executed infinitely often. 
Fair message receipt means that if a computation contains a state in which there is a message in a channel of a node that enables an action in that node, then that action is eventually executed
with the parameters of that message, i.e., the message is eventually
processed. 
Besides these fairness assumptions, we place no bounds on message propagation delay or relative nodes execution speeds, i.e., we allow fully asynchronous computations and non-FIFO message delivery. 
A \emph{computation suffix} is a sequence of computation states past a particular state of this computation. 
In other words, the suffix of the computation is obtained by removing the initial state and finitely
many subsequent states. 
Note that a computation suffix is also a computation.
We say a state $S'$ is reachable from a state $S$ if starting in $S$ there is a sequence of action executions such that we end up in state $S'$.
We use the notion $S < S'$ as a shorthand to indicate that the state $S$ happened chronologically before $S'$.

We consider protocols that do not manipulate the internals of node references. Specifically, a protocol is \emph{compare-store-send} if the only operations that it executes on node references is comparing them, storing them in local memory and sending them in a message. 

In a compare-store-send protocol, a node may learn a new reference of a node only by receiving it in a message. 
A compare-store-send protocol cannot create new references. 
It can only operate on the references given to it.

The overlay network of a set of nodes is determined by their knowledge of each other. 
We say that there is a (directed) \emph{edge} from $a$ to $b$, denoted by $(a,b)$, if node $a$ stores a reference of $b$ in its local memory or has a message in $a.Ch$ carrying the reference of $b$. 
In the former case, the edge is called \emph{explicit}, and in the latter case, the edge is called \emph{implicit}.
Messages can only be sent via explicit edges.
Note that message receipt converts an implicit edge to an explicit edge since the message is in the local memory of a node while it is processed. 
With $NG$ we denote the directed \emph{network (multi-)graph} given by the explicit and implicit edges.
$ENG$ is the subgraph of $NG$ induced by only the explicit edges. 
A \emph{weakly connected component} of a directed graph $G$ is a subgraph of $G$ of maximum size so that for any two nodes $u$ and $v$ in that subgraph there is a (not necessarily directed) path from $u$ to $v$. 
Two nodes that are not in the same weakly connected component are \emph{disconnected}.
We assume that the positions of the processes in the topology are encapsulated in their identifier and that there is a distance measure which is based on the identifiers of the processes and which can be checked locally.
That is, for a given identifier $ID$, each node $u$ can decide for each neighbor $v$ whether $v$ is closer to the node $w$ with $id(w)=ID$ if such a node exists (we also say that $id(v)$ is closer to $ID$ than $id(u)$ or $dist(id(v),ID) < dist(id(u),ID)$).
For a node $u$, we define $R(u,ID)$ as the set containing $u$ and all processes $v$ for which there is a path $Q$ from $u$ to $v$ via explicit edges such that for each edge $(a,b)$ that is traversed in $Q$ it holds that $dist(id(b),ID)< dist(id(a),ID)$.
Furthermore, for a set $U$, we define $R(U,ID) := \bigcup_{u \in U}R(u,ID)$.

In this paper we are particularly concerned with search requests, i.e., \search{v,destID} messages that are routed along $ENG$ according to a given search protocol, where $v$ is the sender of the message and $destID$ is the identifier of a node we are looking for.
We assume that \search{} requests are initiated locally by an (possibly user controlled) application operating on top of the network.
Note that $destID$ does not need to be an id of an existing node $w$, since it is also possible that we are searching for a node that is not in the system. 
If a \search{v,destID} message reaches a node $w$ with $id(w)=destID$, the search request \emph{succeeds}; if the message reaches some node $u$ with $id(u) \ne destID$ and cannot be forwarded anymore according to the given search protocol, the search request \emph{fails}.

\subsection{Problem Statement}

A protocol is \emph{self-stabilizing} if it satisfies the following two properties as long as no transient faults occur.
\begin{description}
\item[\emph{Convergence:}] starting from an arbitrary system state, the protocol is guaranteed to arrive at a legitimate state.
\item[\emph{Closure:}] starting from a legitimate state the protocol remains in legitimate states thereafter.
\end{description}
A self-stabilizing protocol is thus able to recover from transient faults regardless of their nature. 
Moreover, a self-stabilizing protocol does not have to be initialized as it eventually starts to behave correctly regardless of its initial state. 
In \emph{topological self-stabilization} we allow self-stabilizing protocols to perform changes to the overlay network $NG$. 
A legitimate state may then include a particular graph topology or a family of graph topologies.
We are interested in self-stabilizing protocols that stabilize to \emph{static topologies}, i.e., in every computation of the protocol that starts from a legitimate state, $ENG$ stays the same, as long as the node set stays the same.

In this paper we are not focusing on building a self-stabilizing protocol for a particular topology.
Instead we are interested in providing a reliable protocol for searching in a wide range of topologies that fulfill certain requirements. 
Traditionally, search protocols for a given topology were only required to deliver the search messages reliably once a legitimate state has been reached.
However, it is not possible to determine when a legitimate state has been reached.
Furthermore, searching reliably during the stabilization phase is much more involved.
We say a self-stabilizing protocol satisfies \emph{monotonic searchability} according to some search protocol $R$ if it holds for any pair of nodes $v,w$ that once a \search{v,id(w)} request (that is routed according to $R$) initiated at time $t$ succeeds, any \search{v,id(w)} request initiated at a time $t' > t$ will succeed.
We do not mention $R$ if it is clear from the context.
A protocol is said to satisfy \emph{non-trivial} monotonic searchability if (i) it satisfies monotonic searchability and (ii) every computation of the protocol contains a suffix such that for each pair of nodes $v,w$, \search{v,id(w)} requests will succeed if there is a path from $v$ to $w$ in the target topology.
Throughout the paper we will only investigate non-trivial monotonic searchability.
Consequently, whenever we use the term monotonic searchability in the following, we implicitly refer to non-trivial monotonic searchability.

A \emph{message invariant} is a predicate of the following form:
If there is a message $m$ in the incoming channel of a node, then a logical predicate $P$ must hold.
A protocol may specify one or more message invariants.
An arbitrary message $m$ in a system is called \emph{corrupted} if the existence of $m$ violates one or multiple message invariants.
A state $S$ is called \emph{admissible} if there are no corrupted messages in $S$.
We say a (self-stabilizing) protocol \emph{admissibly satisfies} a predicate $P$ if the following two conditions hold: (i) the predicate is satisfied in all computation suffixes of the protocol that start from admissible states, and (ii) every computation of the protocol contains at least one admissible state.
A protocol \emph{unconditionally satisfies} a predicate if it satisfies this predicate starting from any state.

The following was proven in~\cite{SetzerSSS14}:
\begin{lemma}\label{lem:admissible_message_necessary_for_monotonic_searchability}
No self-stabilizing compare-store-send protocol can unconditionally satisfy monotonic searchability.
\end{lemma}

Consequently, to prove monotonic searchability for a protocol (according to a given search protocol $R$) it is sufficient to show that: (i) in every computation of the protocol that starts from an admissible state, every state is admissible, (ii) in every computation of the protocol there is an admissible state, and (iii) the protocol satisfies monotonic searchability according to $R$ in every computation that starts from an admissible state.
Note that we have not defined any invariants yet and it is possible to pick invariants such that the set of admissible states equals the set of legitimate states, in which the problem becomes trivial.
However, for the invariants we provide, any initial topology can be an admissible state.
In particular, as long as no corrupt messages are initially in the system, our protocols satisfy monotonic searchability throughout the computation.

We will show that a broad class of existing self-stabilizing protocols can be transformed to satisfy monotonic searchability.
More specifically, we will consider protocols that fulfill the \emph{\mdl property}, i.e., for any action $a$ of the protocol it holds that
(i) a node $u$ executing action $a$ will always keep a reference of another node $v$ in its local memory if an edge $(u,v)$ is part of the final topology, and
(ii) if a node $u$ executing action $a$ in state $S$ decides not to keep a reference of another node $v$ in its local memory, every other action of the protocol executed by $u$ in a subsequent state will decide to not keep the reference of $v$, and
(iii) a node $u$ executing action $a$ decides \emph{deterministically} and solely based on its \emph{local} memory whether to send and where to send the reference of $v$,
(iv) in every legitimate state, for every reference of a node $v$ contained in a message $m$ in the channel of a node $u$ (i.e., for any implicit edge $(u,v)$), there are fixed cycle-free paths $(u = u_1, u_2, \dots, u_k)$ such that $u_i$ sends the reference of $v$ to $u_{i+1}$, and $u_k$ has an explicit edge $(u_k,v)$ (note that there is only one path if the reference is never duplicated), i.e., the reference of $v$ is forwarded along fixed paths until it finally fuses with an existing reference, and
(v) any stable edge that may be traversed by the search protocol and any implicit edge that results from a delegation of an explicit edge that may have been traversed by the search protocol, is only delegated to a node whose distance (in the underlying distance metric) is closer to the target than the corrent node.
 Informally speaking, the first two properties imply that the protocol \emph{monotonically} converges to its desired topology, since edges of the topology are always kept and edges that are not part of the topology are obviated over time.
 The last property implies that in legitimate states, all implicit edges will eventually merge with explicit edges.
 Note that the \mdl property is generally not a severe restriction. 
 Most existing protocols that stabilize to static topologies naturally fulfill this property (see Section~\ref{sec:examples}).

In addition to the \mdl property and the assumption that the final topology is static, we have one more condition on the topologies and their distance measures.
The generic search protocol we will use to achieve monotonic searchability assumes that in the target topology for every pair of nodes $u,v$ within the same connected component, there is a path of explicit edges from $u$ to $v$ with the property that each edge on the path strictly decreases the distance to $v$ (i.e., for each edge $(a,b)$ that is traversed in the path, $dist(id(b),id(v))< dist(id(a),id(v))$).
Note that many topologies naturally fulfill this property (in particular, whenever the distance is defined as the number of nodes on a shortest path).

\subsection{Our contribution}
To the best of our knowledge, we are the first to solve the problem of searching reliably during the stabilization phase in self-stabilizing topologies.
Although routing with a low dilation is a major motivation behind the use of overlay topologies, prior to this work, one could not rely on the routing paths in such topologies\footnote{Note that \cite{SetzerSSS14} did solve the problem of monotonic searchability for the list, but the list has a worst-case routing time of $\Omega(n)$, thus not offering a low dilation.}: In previous approaches, it can happen that a node $u$ is able to send a message to a node $v$, while it is unable to do so in a later state, only because the system has not stabilized yet (which is not locally detectable by the nodes).
In our solution, once a search message from a node $u$ has successfully arrived at a node $v$, every further search message from $u$ to $v$ will also arrive at its destination, regardless of whether the system has fully stabilized or not.

We present a universal set of primitives for manipulating edges that protocols should use and a simple generic search protocol, which together satisfy monotonic searchability.
Moreover, we provide a general description of how a broad class of self-stabilizing protocols for overlay networks can be transformed such that they use these primitives, thus satisfying monotonic searchability afterwards.

Our results of Section~\ref{sec:preliminaries} may be of independent interest, where we reinvestigate the fundamental primitives for manipulating edges introduced in~\cite{KoutsopoulosSS15} and strengthen the results concerning the universality of these primitives.

\subsection{Outline of the paper}
The paper is organized as follows:
We give a short overview of related work in Section~\ref{sec:related_work}.
In Section~\ref{sec:preliminaries}, we describe primitives for manipulating edges introduced in~\cite{KoutsopoulosSS15} that are the base for our new primitives for searchability.
In Section~\ref{sec:new_primitives}, we describe these new primitives that one should use in a protocol to enable monotonic searchability and show that conventional protocols for the self-stabilization of a topology can be transformed into ones using the new primitives.
In Section~\ref{sec:generic_search_protocol}, we present a generic search protocol according to which such protocols satisfy monotonic searchability and prove its correctness.
In Section~\ref{sec:examples}, we give examples on how to apply our results to existing topologies, namely the list and the \skipp graph.
A conclusion and outlook (Section~\ref{sec:conclusion}) mark the end of the paper.

\section{Related work}\label{sec:related_work}
The idea of self-stabilization in distributed computing was introduced in a classical paper by E.W. Dijkstra in 1974~\cite{Dijkstra74}, in which he investigated the problem of self-stabilization in a token ring. 
In order to recover certain network topologies from any weakly connected state, researchers started with simple line and ring networks (e.g.,~\cite{ShakerR05,self-stabilizing-list,self-stabilizing-list2}).
Over the years more and more network topologies were considered, ranging from skip lists and skip graphs~\cite{corona,JRSST09}, to expanders~\cite{DolevT2013}, 
hypertrees and double-headed radix trees~\cite{DolevK08,AspnesW07}, and small-world graphs~\cite{KniesburgesKS12}.
Also a universal algorithm for topological self-stabilization is known~\cite{DBLP:journals/tcs/BernsGP13}.

In the last 20 years many approaches have been investigated that focus on maintaining safety properties during convergence phase of self-stabilization, e.g. snap-stabilization~\cite{DBLP:journals/dc/BuiDPV07,DBLP:journals/jpdc/DelaetDNT10}, super-stabilization~\cite{DBLP:journals/cjtcs/DolevH97}, safe convergence~\cite{DBLP:conf/ipps/KakugawaM06} and self-stabilization with service guarantee~\cite{DBLP:conf/europar/JohnenM10}.
A protocol is snap-stabilizing if it always behaves according to its specification independent of its initial configuration. 
Snap-stabilization has a user-centric safety property (whereas the other approaches are system-centric), i.e. it ensures that the answer to a properly initiated user request by the protocol is correct. 
Safe convergence ensures that (1) the system quickly converges to a safe configuration, and (2) the safety property stays satisfied during the stabilization under protocol actions. 
However, external disruptions are not handled in safe convergence. 
A super-stabilizing protocol guarantees that (1) starting from a legitimate configuration, a safety property is preserved after only one specific topology change, and (2) the safety property
is maintained during recovering to a legitimate configuration assuming that
no more topology change occurs during stabilization phase. 
Self-stabilization with service guarantee provides and maintains the safety property even before
stabilization, unlike super-stabilization. 

Closest to our work is the notion of \emph{monotonic convergence} by Yamauchi and Tixeuil~\cite{YamauchiT10}. A self-stabilizing protocol is monotonically converging if every change done by a node $p$ makes the system approach a legitimate state and if every node changes its output only once. 
The authors investigate monotonically converging protocols for different classical distributed problems (e.g., leader election and vertex coloring) and focus on the amount of non-local information that is needed for them.

Research on monotonic searchability was initiated in~\cite{SetzerSSS14}, in which the authors proved that it is impossible to satisfy monotonic searchability if corrupted messages are present. In addition, they presented a self-stabilizing protocol for the line topology that is able to satisfy monotonic searchability.

\section{Primitives for Topology Maintenance}
\label{sec:preliminaries}
An important property for every overlay management protocol is that weak connectivity is never lost by its own actions. 
Therefore, it is highly desirable that every node only executes actions that preserve weak
connectivity. 
Koutsopoulos et al.~\cite{KoutsopoulosSS15} introduced the following four primitives for manipulating edges in an overlay network.

\begin{description}
\item[Introduction]  If a node $u$ has a reference of two nodes $v$ and $w$ with $v \neq w$, $u$ \emph{introduces} $w$ to $v$ if $u$ sends a message to $v$ containing a reference of $w$ while keeping the reference.
\item[Delegation] If a node $u$ has a reference of two nodes $v$ and $w$ s.t. $u,v,w$ are all different, then $u$ \emph{delegates} $w$'s reference of $v$ if $u$ sends a message to $v$ containing a reference of $w$ and deletes the reference of $w$.
\item[Fusion] If a node $u$ has two references $v$ and $w$ with $v=w$, then $u$ \emph{fuses} the two references if it only keeps one of these references.
\item[Reversal] If a node $u$ has a reference of some other node $v$, then $u$ \emph{reverses} the connection if it sends a reference of itself to $v$ and deletes its reference of $v$.
\end{description}

\begin{figure}[htb]
\centering 
\subfloat[Introduction primitive]{
\includegraphics[width=0.4\textwidth]{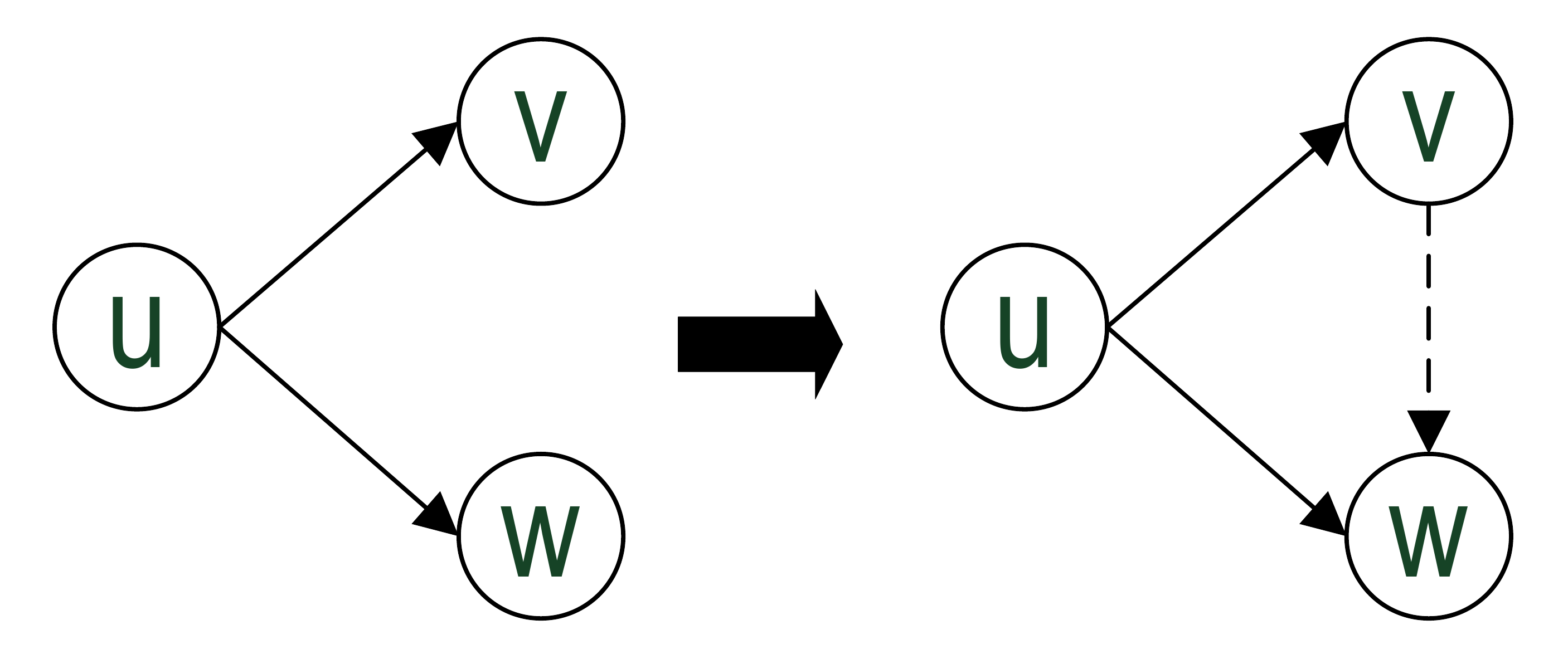}
\label{fig:prim-intro}
}
\quad
\
\subfloat[Delegation primitive]{
\includegraphics[width=0.4\textwidth]{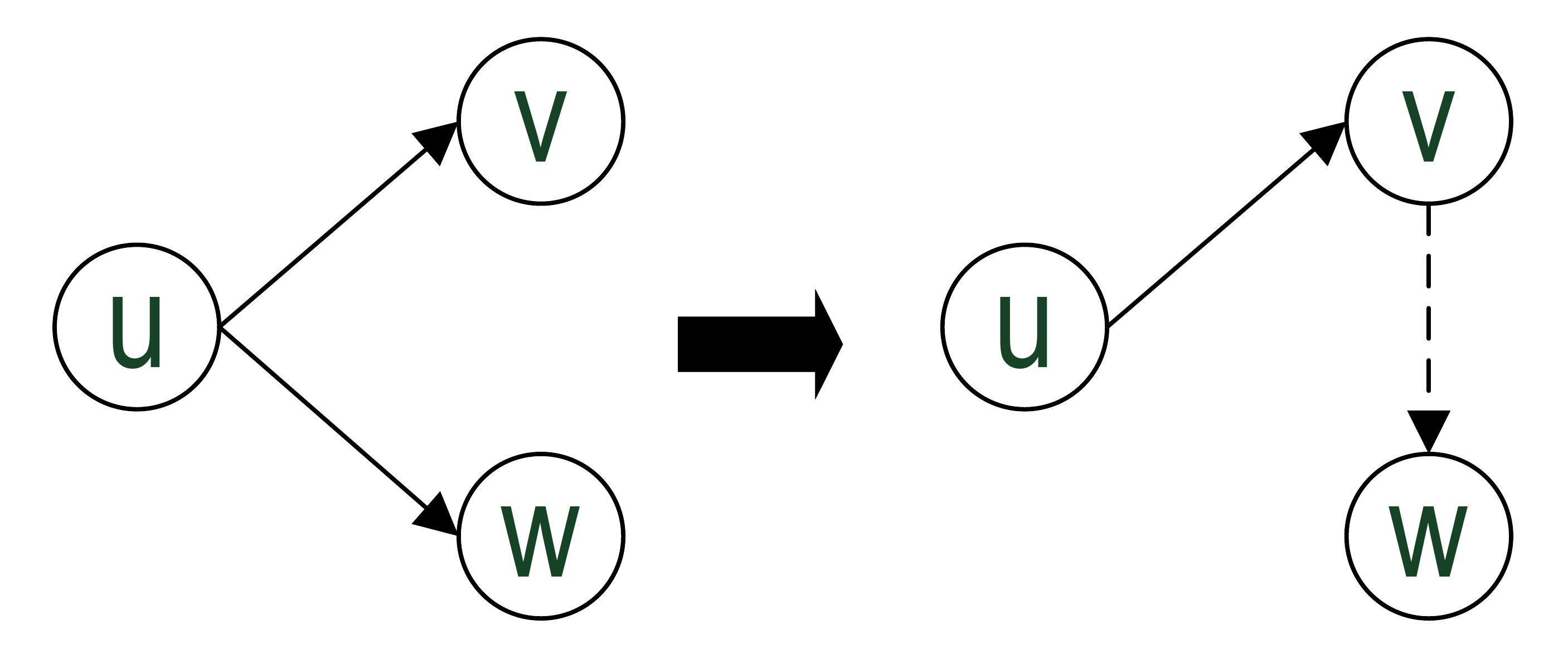}
\label{fig:prim-dele}
}
\\
\subfloat[Fusion primitive]{
\includegraphics[width=0.4\textwidth]{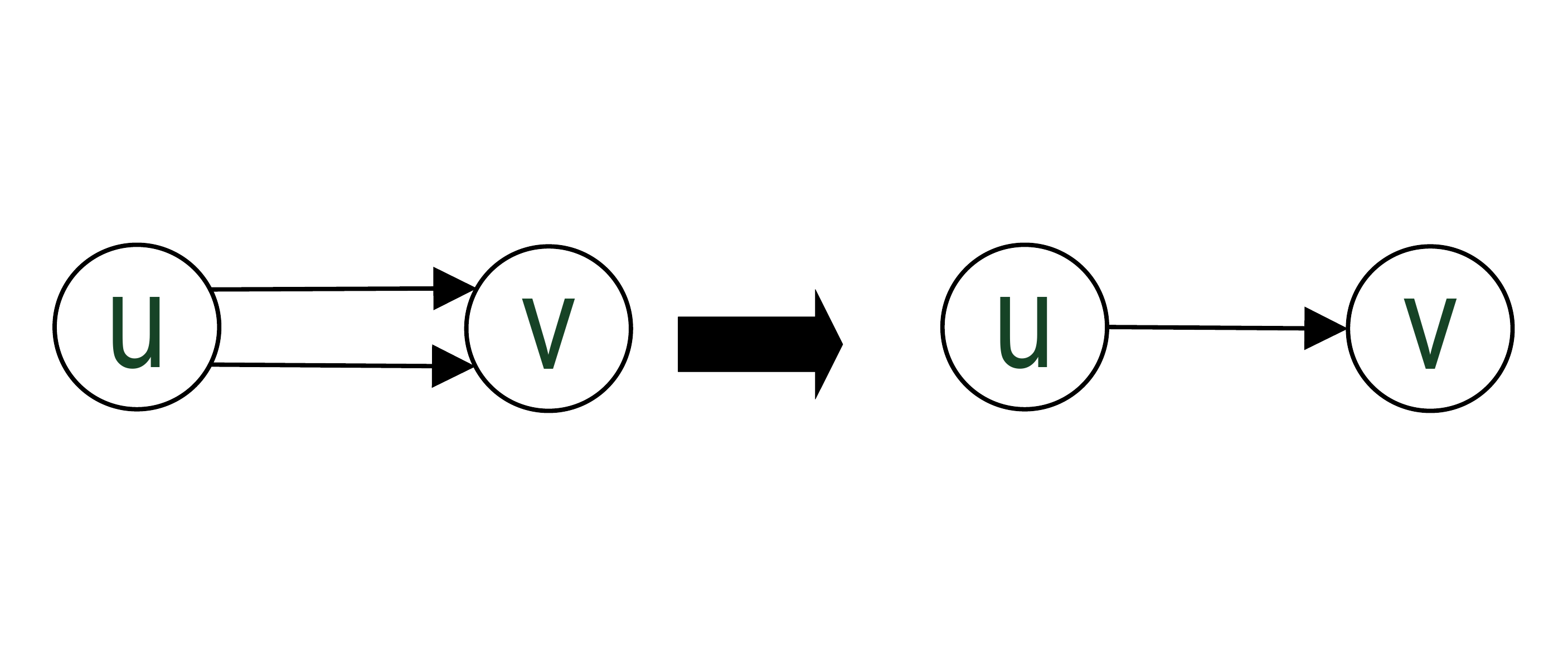}
\label{fig:prim-fuse}
}
\quad
\
\subfloat[Reversal primitive]{
\includegraphics[width=0.4\textwidth]{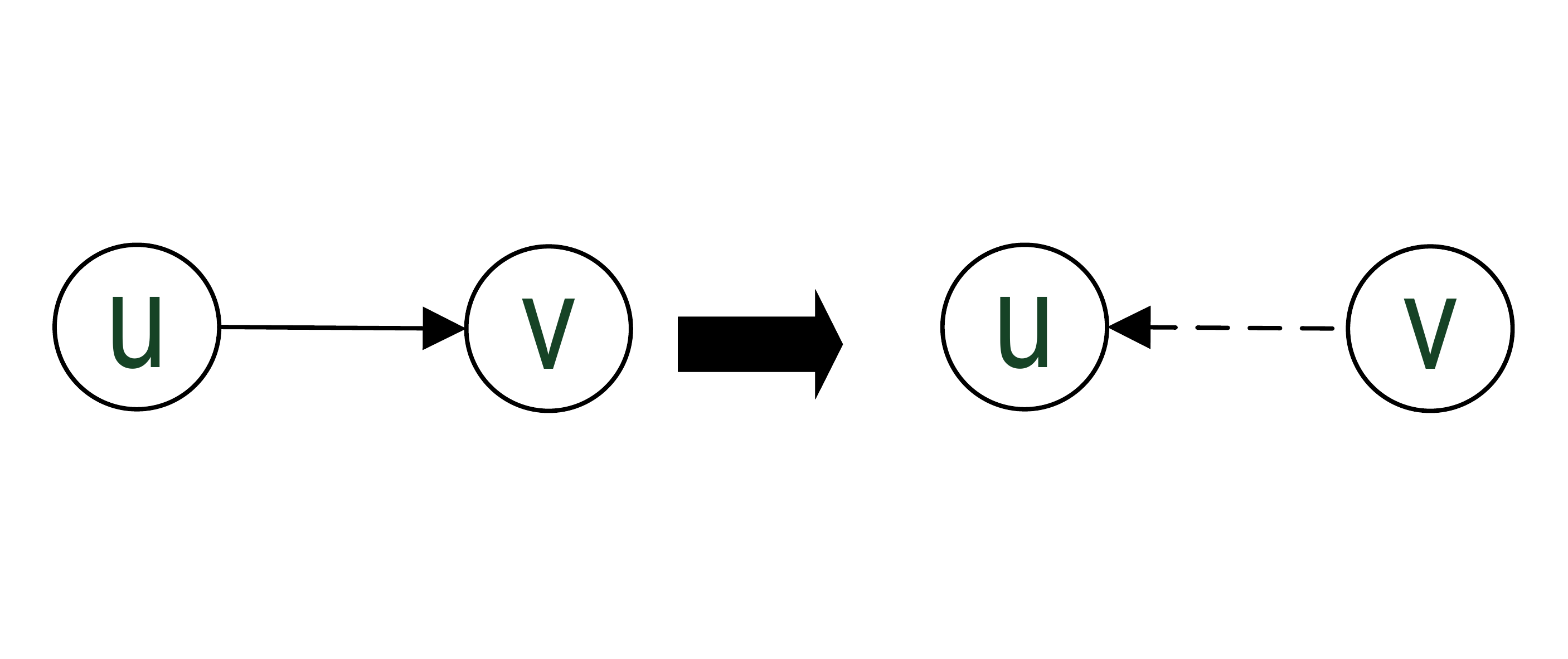}
\label{fig:prim-reverse}
}
\caption{The four primitives in pictures. Explicit edges are drawn solid; implicit edges are dashed.}\label{fig:primitives}
\end{figure}

Note that the four primitives can be executed locally by every node in a wait-free fashion.
Furthermore, for the Introduction primitive, it is possible that $w=u$, i.e., $u$ introduces itself to $v$.
The authors show that these four primitives (which are visualized in Figure~\ref{fig:primitives}) are safe in a sense that they preserve weak connectivity (as long as there is no fault).
This implies that \emph{any} distributed protocol whose actions can be decomposed into these four primitives is guaranteed to preserve weak connectivity.

We define $\mathcal{IDF}$ as the set containing the first three primitives: Introduction, Delegation and Fusion.
Let \pidf denote the set of all distributed protocols where all interactions between processes can be decomposed into the primitives of $\mathcal{IDF}$.
According to~\cite{KoutsopoulosSS15} these protocols even preserve strong connectivity in a sense that for any pair of nodes $u,v$ with a directed path in $NG$ there will always be a directed path from $u$ to $v$ in $NG$.
To the best of our knowledge, all self-stabilizing topology maintenance protocols proposed so far (such as the list~\cite{ShakerR05,self-stabilizing-list,self-stabilizing-list2}, the Delaunay graph~\cite{JacobRSS2012}, etc.) satisfy this property.
Moreover, in~\cite{KoutsopoulosSS15}, the four primitives were shown to be \emph{universal}, i.e. the primitives allow one to get from any weakly connected graph $G=(V,E)$ to any other weakly connected graph $G'=(V,E')$ for $NG$.
In fact, only the first three primitives  (i.e., $\mathcal{IDF}$) are necessary to get from any weakly connected graph to any \emph{strongly} connected graph, which is sufficient in our case (\cite{KoutsopoulosSS15} denote this by \emph{weak universality}).
Note that the notion of universality for a set of primitives is not constructive, i.e., only \emph{in principle} the primitives allow one to get from any weakly connected graph to any other weakly connected graph. 
We strengthen the results concerning universality of the primitives with the following theorem.
\begin{theorem}
\label{thm:realUnviversality}
 Any compare-store-send protocol that self-stabilizes to a static strongly-connected topology and preserves weak connectivity can be transformed such that the interactions between nodes can be decomposed into the primitives of $\mathcal{IDF}$.
\end{theorem}

The rest of this section is dedicated to proving Theorem~\ref{thm:realUnviversality}.
We say that a node $u$ \emph{deletes} its reference of another node $v$, if $u$ executes an action such that $(u,v)$ is explicit, not a multi-edge, and $u$ removes the reference from its local memory without sending it to another node. 
Before we are able to prove Theorem~\ref{thm:realUnviversality}, we need the following lemma.
\begin{lemma}
\label{lem:noDeletion}
Any compare-store-send protocol that self-stabilizes to a strongly-connected topology and contains an action such that a node $u$ deletes a reference does not preserve weak connectivity.
\end{lemma}
\begin{proof}
Assume for contradiction that the protocol preserves weak connectivity and there is an action in which $u$ deletes a reference.
Consider the left graph depicted in Figure~\ref{fig:deletionCounterExample} and assume that $u$ aims at deleting the edge $(u,v)$. It is possible that $u$ deletes the reference of $v$ immediately.
However, since the protocol presumably preserves weak connectivity, $u$ may not delete the edge immediately, but is allowed to perform other actions.
However, since the graph is still connected without $(u,v)$, $u$ will eventually decide that it can delete $(u,v)$ safely and execute the action that ultimately deletes the reference of $v$.
Let $S'$ be the \emph{internal state} of $u$ before it deletes the edge, i.e., the values of all variables of $u$ and messages in $u.Ch$.
  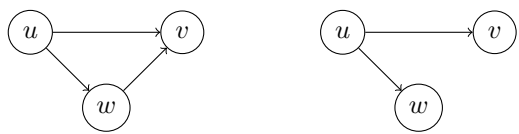
\begin{figure}[ht]
    \centering
 	\begin{tikzpicture}
		\node(u)[circle,draw=black] at (0,0) {$u$};
		\node(v)[circle,draw=black] at (1,-1) {$w$};
		\node(w)[circle,draw=black] at (2,0) {$v$};
		\draw [->] (u) -- (w);
		\draw [->] (u) -- (v);
		\draw [->] (v) -- (w);
	\end{tikzpicture}
	\qquad \qquad
	 \begin{tikzpicture}
		\node(u)[circle,draw=black] at (0,0) {$u$};
		\node(v)[circle,draw=black] at (1,-1) {$w$};
		\node(w)[circle,draw=black] at (2,0) {$v$};
		\draw [->] (u) -- (w);
		\draw [->] (u) -- (v);
	\end{tikzpicture}
	\caption{Graph instances for the proof of Lemma~\ref{lem:noDeletion}.}\label{fig:deletionCounterExample}
 \end{figure}
We construct a new system state by taking the right graph depicted in Figure~\ref{fig:deletionCounterExample} and setting the internal state of $u$ to $S'$.
Naturally, this is a valid initial state for a computation.
Since the internal state of $u$ has not changed it still makes the decision to delete its reference of $v$.
However, this disconnects $u$ and $v$, thereby weak connectivity is lost which is a contradiction.
\end{proof}

Now we can prove Theorem~\ref{thm:realUnviversality}.

\begin{proof}[Proof of Theorem~\ref{thm:realUnviversality}]
Let $A$ be a compare-store-send protocol that self-stabilizes to a strongly-connected topology and preserves weak connectivity and let $u$ be a node that acts according to $A$.
Since $A$ is a compare-store-send protocol we focus on the actions of $A$ that specifically handle references of nodes.

At first consider all actions of $A$ that are executed because a local predicate becomes true, i.e., $u$ does not receive any message.
Node $u$ has three options: (i) it does not interact with its references, or (ii) it sends one or multiple references to one or multiple nodes, or (iii) it deletes one or multiple references.
In case (i) there is no interaction between nodes.
In case (ii) $u$ sends references and either keeps them in its memory or does not keep them.
The first subcase can be transformed such that the Introduction primitive is used.
For the second subcase we can use the Delegation primitive.
Case (iii) is not allowed due to Lemma~\ref{lem:noDeletion}.

Next consider all actions of $A$ that are triggered by a message that contains at least one reference, i.e., the message is received by $u$ and has to be processed.
Node $u$ has multiple options for each reference in the received message: (i) keep the reference, and/or (ii) send the reference to one or multiple nodes, or (iii) delete the reference (i.e., neither save the reference nor send it to another node).
Case (i) is either the Fusion primitive (in case $u$ already has a reference of that node) or does not need to be handled by a primitive since keeping an edge does not change the network.
Case (ii) is again Introduction or Delegation, depending on whether $u$ also saves the reference.
Similar to the previous case, case (iii) is not allowed.
Moreover, $u$ has the aforementioned option of sending one or multiple references of nodes that are not in the message to one or multiple nodes.

Finally, consider all actions of $A$ that are triggered by a message that does not contain any references.
This case basically reduces to the case in which a predicate triggers action execution because node $u$ has the same options for interaction with its references.
\end{proof}

\section{Primitives for Monotonic Searchability}\label{sec:new_primitives}
Although the primitives of \cite{KoutsopoulosSS15} are general enough to construct any conceivable overlay, they do not inherently satisfy monotonic searchability. 
This is due to the fact that the Delegation primitive replaces an explicit edge $(u,v)$ by a path $(u,w,v)$ consisting of an explicit edge $(u,w)$ and an implicit edge $(w,v)$ and thus a search message from $u$ to $v$ issued after the delegation may be processed by $w$ before there is a path from $w$ to $v$ via explicit edges, causing the search message to fail (even though an earlier message sent while $(u,v)$ was still an explicit edge was delivered successfully).
Consequently, we are going to introduce a new set of primitives that enables monotonic searchability.
We say a set of primitives is \emph{search-universal} according to a set of Invariants \invar if the following holds:
\begin{enumerate}
 \item the set of primitives is weakly universal,
 \item starting from every state in which the invariants in \invar hold, for every pair of nodes $u$ and $v$ as soon as there is a path via explicit edges from $u$ to $v$, there will be a path via explicit edges from $u$ to $v$ in every subsequent step. 
\end{enumerate}
We are now going to introduce a modified set of primitives that are search-universal.
Moreover, we will show that these new primitives are also general enough to cover all self-stabilizing protocols that can be built by the original primitives.
Consequently, we ultimately aim at a result similar to Theorem~\ref{thm:realUnviversality} for the new primitives.

Remember that we assume the \mdl property.
Therefore, in every fixed state $S$ in every execution of a self-stabilizing protocol, each node $u$ can divide its explicit edges into two subsets: the \emph{stable edges} and the \emph{temporary edges} (not to be confused with implicit edges).
The first set contains those explicit edges that $u$ wants to keep, given its current neighborhood in $S$; the second set holds the explicit edges that are not needed from the perspective of $u$ in $S$.
Note that the set of temporary edges can also be the empty set.

For the new primitives, a node does not only store references of its neighbors, but additionally stores sequence numbers for every reference in its local memory, i.e., every node $u$ stores for each neighbor $v$ an entry $u.eseq[id(v)]$ (or $u.eseq[v]$, in short).
Consider the following primitives:
\begin{description}
 \item[Introduction] If a node $u$ has a reference of two nodes $v$ and $w$ with $v \neq w$, $u$ \emph{introduces} $w$ to $v$ if $u$ sends a message to $v$ containing a reference of $w$ while keeping the reference.
 \item[Safe-Delegation] Consider a node $u$ that has references of two different nodes $v$ and $w$. 
	In order to perform \emph{Safe-Delegation}, $u$ has to distinguish between $(u,w)$ being implicit or temporary.
 
If $(u,w)$ is an implicit edge, it is delegated as in the original delegation primitive (we will later refer to this case as an \emph{implicit delegation} or \impldelegate{} to avoid confusion with the original primitives).
If $(u,w)$ is a temporary edge, it can only be delegated to a node $v$ if $(u,v)$ is a stable edge.
 Whenever an explicit edge $(u,w)$ is to be delegated to another node $v$, $u$ sends a $\delegateREQ{u,w,eseq}$ message to $v$, where $eseq = u.eseq[w]$.
 Additionally, it sets $u.eseq[v]$ to $max\{u.eseq[v], u.eseq[w] + 1\}$.
 Any node $v$ that receives a $\delegateREQ{u,w,eseq}$ message, adds $(v,w)$ to its set of explicit edges (if it does not already exist), sets $v.eseq[w]$ to $\max\{v.eseq[w],eseq + 1\}$ and sends a $\delegateACK{w,eseq}$ message back to $u$. 
 Upon receipt of this message, $u$ checks whether $eseq = u.eseq[w]$ and whether $(u,w)$ is actually a temporary edge (note that the last check is necessary to handle corrupt initial states).
 If both conditions hold, $u$ removes the temporary explicit edge to $w$ and sends an \impldelegate{w} message to one of its neighbors. 
 Otherwise, $u$ simply acts as it would upon receipt of an \impldelegate{w} message.
    \item[Fusion] If a node $u$ has two references $v$ and $w$ with $v=w$, then $u$ \emph{fuses} the two references if it only keeps one of these references.
	Note that when a node $u$ receives a \delegateREQ{v,w,eseq} message and already stores a reference of $w$, it also behaves as described in the Safe-Delegation primitive.
\end{description}

We define $\mathcal{ISF}$ as the set containing the three primitives Introduction, Safe-Delegation and Fusion.
Throughout the paper we assume that  \delegateREQ{} and \delegateACK{} messages are only sent in the Safe-Delegation primitive.
Analogous to \pidf, let \pisf denote the set of all distributed protocols where all interactions between processes can be decomposed into the primitives of $\mathcal{ISF}$.
Likewise to the \mdl property, we say that a protocol fulfills the \emph{stable} \mdl property, if the protocol fulfills the \mdl property with respect to stable explicit edges.
More specifically, for any action $a$ of the protocol it holds that a node $u$ executing action $a$ will always keep a reference of another node $v$ in its local memory (i.e., the stable edge $(u,v)$) if an edge $(u,v)$ is part of the final topology, and if a node $u$ executing action $a$ in some state $S$ decides to not keep a reference of another node $v$ in its local memory (i.e., the temporary or implicit edge $(u,v)$), every other action of the protocol executed by $u$ in any state $S' > S$ will decide to not keep the reference of $v$.

\subsection{Universality of the new primitives}
To show that our primitives are search-universal we first show that they are weakly universal.
\begin{lemma}\label{lem:pms_is_weakly_universal}
 $\mathcal{ISF}$ is weakly universal.
\end{lemma}
\begin{proof}
The idea of this proof similar to that of the proof of (weak) universality for the original primitives~\cite{KoutsopoulosSS15}.
However, it is a bit more involved because of the Safe-Delegation primitive.
In general, we will present a simple general strategy of how to transform an arbitrary weakly connected graph $G=(V,E)$ into any other strongly connected graph $G'=(V,E')$. 
At first, we use the Introduction primitive to transform $G$ into a clique, i.e. every node continuously introduces all neighbors (including itself) to each other. 
Obviously, $O(\log n)$ communication rounds are sufficient to build a clique as the distances between the nodes halve in each round of introduction.

Next we show that by using Safe-Delegation and Fusion, one can transform the clique to $G'=(V,E')$.  
In general, each node $u$ always keeps all references of all other nodes $w$ with $(u,w) \in E'$.
Furthermore, whenever a node $u$ has two references of the same node $w$, it fuses them.
Consider the following approach:
For every edge $(u,w)$ that is in the clique but not in $E'$, each node $u$ uses the Safe-Delegation primitive to delegate the reference of $w$ to the neighbor $v$ that is the next node on the shortest path from $u$ to $w$ in $E'$ (in the following, we use the notation $P_{E'}(u,w)$ for this path).
That is, $u$ sends a \delegateREQ{u,w,eseq} message to $v$ and when it receives a \delegateACK{w,eseq} message, it no longer stores the reference of $w$, but delegates it to $v$ (which is still a neighbor of $u$ by what we said before) via an \impldelegate{w,eseq} message.
We show that by this procedure, all edges $u,v \notin E'$ will eventually vanish.

We define the potential $X := \sum_{u \in V} X(u)$, with $X(w) := \max_{u \in V: (u,w) \in NG \setminus E'}|P(u,w)|$ for a node $w$.
Note that no node ever delegates an edge $(x,z)$ to a node $y$ such that $|P(y,z)| > |P(x,z)|$.
Thus, $X$ never increases.
We show that $X$ is monotonically decreasing.
For an arbitrary node $w$, consider a node $u \in argmax_{u' \in V: (u',w) \in E}(|P(u',w)|)$ and let $d := |P(u,w)|$.
According to the procedure we describe, $u$ sends a \delegateREQ{u,w,eseq} message to $v$, the next node on $P_{E'}(u,w)$.
By the fair message receipt assumption, $v$ will receive this message, act according to the Safe-Delegation primitive and send back a \delegateACK{w,eseq} message to $u$ which will ultimately be received by $u$.
Note that $u.eseq[w] = eseq$ in this state because the only occasion at which $u$ would have increased $u.eseq[w]$ is that it received a \delegateREQ{x,w,eseq'} message for some $x$ and some number $eseq'$, but this would imply $(x,w) \in NG$ at the time of sending of that message and $|P(x,w)| > P|(u,w)|$ yielding a contradiction to the definition of $u$.
Thus, $u$ will no longer store the reference of $w$ and forward it via an \impldelegate{w} message to $v$.
As soon as this has happened for every $u'$ with $|P(u',w)|=d$, $X$ decreases.
By induction we have that $X$ will eventually be $0$, which finishes our proof.
\end{proof}

In order to enable monotonic searchability, we define the following two \textbf{message invariants}:
\begin{enumerate}
 \item If there is a \delegateREQ{u,w,eseq} message in $v.Ch$, then there exists a path $P = (u=x_1, x_2, \dots, x_k = v)$ that does not contain $(u,w)$ and for every $1 \leq i < k$, $x_i.eseq[x_{i+1}] > u.eseq[w]$, or $u.eseq[w] > eseq$.\label{Inv:DelegateReq}
 \item If there is a \delegateACK{w,eseq} message in $u.Ch$, then there exists a path $P = (u=x_1, x_2, \dots, x_k = w)$ that does not contain $(u,w)$ and for every $1 \leq i < k$, $x_i.eseq[x_{i+1}] > u.eseq[w]$, or $u.eseq[w] > eseq$.\label{Inv:DelegateAck}
\end{enumerate}
Intuitively, Invariant~1 states that whenever a node $v$ has a \delegateREQ{u,w,eseq} message in $v.Ch$ (i.e., node $u$ asked $v$ to establish the edge $(v,w)$ such that it may remove its own $(u,w)$ edge), then there is a path from $u$ to $v$ that does not use the edge $(u,w)$.
Invariant~2 states that whenever a node $u$ has a \delegateACK{w,seq} message in $u.Ch$ (i.e., some other node $v$ which $u$ asked to establish the edge $(v,w)$ has already done so), then there is a path from $u$ to $w$ that does not use the edge $(u,w)$.
However, both statements only need to hold if the value of \textsc{eseq} indicates that the messages belong to a current safe-delegation, i.e., if $u.seq[w] > eseq$, the \delegateREQ{} or \delegateACK{} message can be ignored.

We define the predicate $E(u,v)$ to be true if and only if there exists a directed  path from $u$ to $v$ via explicit edges.
In order to show search-universality, we prove the following lemma.
\begin{lemma}\label{lem:once_A_and_B_hold_then_in_every_subsequent_step_as_well}
 Consider a computation of a protocol $P \in \pisf$ that fulfills the stable \mdl property. If there is a state $S$ such that Invariant~\ref{Inv:DelegateReq} and Invariant~\ref{Inv:DelegateAck} hold, then they will hold in every subsequent state.
 Additionally, for every state $S' \geq S$ it holds that if $E(u,v) \equiv TRUE$ in $S'$, then $E(u,v) \equiv TRUE$ in every state $S'' \geq S'$.
\end{lemma}
\begin{proof}
  We show that if the statements hold in state $S$, then they hold in the subsequent state $S^*$, from which the claim follows by induction.
  
  Since in state $S$ Invariant~\ref{Inv:DelegateReq} holds we can conclude that:
  \begin{itemize}
   \item When a new \delegateREQ{u,w,eseq} message is sent out from node $u$ to node $v$, there has to be an explicit edge $(u,v)$, and $u.eseq[v]$ is set to at least $u.eseq[w] + 1$.
   \item The invariant can become false if an edge $(x,y)$ on the only remaining path $P$ from $u$ to $v$ that does not contain $(u,w)$ is removed or becomes implicit.
   According to the primitives stated above, this can only happen if $x$ receives a $\delegateACK{y,eseq'}$ message and $eseq' = x.eseq[y]$.
   For this message, the second invariant has to hold, yielding the existence of a path $Q = (x = x_1, x_2, \dots, x_k = y)$ even after the removal of $(x,y)$ such that for every $1 \leq i < k :$ $x_i.eseq[x_{i+1}] > x.eseq[y]$.
   Notice that since $P$ was the only remaining path from $u$ to $v$ not via $(u,w)$, the first invariant implies that $x.eseq[y] > u.eseq[w]$.
   Thus, there is a path $R = (u = y_1, y_2, \dots, w = y_l)$, which is $P$ with $(x,y)$ being replaced by $Q$, such that for all $1 \leq i < l : $ $x_i.eseq[x_{i+1}] > u.eseq[w]$ holds.
   In particular, this implies that $(u,w)$ is not contained in $R$.
   Thus, the first invariant still holds.
  \end{itemize}
  Considering Invariant~\ref{Inv:DelegateAck} we can conclude that:
  \begin{itemize}
   \item A new \delegateACK{w,eseq} message is sent to a node $u$ only as a response to a \delegateREQ{u,w,eseq} message received by a node $v$.
   According to Invariant~\ref{Inv:DelegateReq}, there exists a path $P = (u=x_1, x_2, \dots, x_k = v)$ that does not contain $(u,w)$ and for every $1 \leq i < k : $ $x_i.eseq[x_{i+1}] > u.eseq[w]$ (or $u.eseq[w] > eseq$, in which case we are done).      
   Since $v$ adds a new explicit edge $(v,w)$ when it sends out the \delegateACK{w,eseq} message and also ensures that $v.eseq[w] > eseq$ (by adjusting $v.eseq[w]$), the invariant still holds afterwards.
   \item The invariant can become false if an edge $(x,y)$ on the only remaining path $P$ from $u$ to $v$ that does not contain $(u,w)$ is removed or becomes implicit.
    This case is analogous to the second case of Invariant~\ref{Inv:DelegateReq}.
  \end{itemize}
Consequently both invariants hold in $S^*$.
  
  Since an explicit edge only turns implicit (i.e., a node removes the explicit edge from $NG$ and creates a new implicit edge to maintain connectivity) if a \delegateACK{} message is received, Invariant~\ref{Inv:DelegateAck} ensures that if $E(u,v) \equiv TRUE$ in state $S' \geq S$, $E(u,v) \equiv TRUE$ in every state $S'' \geq S'$. 
\end{proof}

Lemma~\ref{lem:pms_is_weakly_universal} and Lemma~\ref{lem:once_A_and_B_hold_then_in_every_subsequent_step_as_well} imply the following corollary:
\begin{corollary}
 $\mathcal{ISF}$ is search-universal according to Invariant~1 and Invariant~2.
\end{corollary}

We conclude this section by showing that a protocol $A \in \pidf$ that fulfills the \mdl property and self-stabilizes to some topology can be transformed into a protocol $B \in \pisf$ that fulfills the stable \mdl property and for which it holds that in every computation of $B$ there is a state in which Invariant~1-2 hold.
   
\begin{theorem}\label{thm:conventional_protocols_transformable_such_that_invarA_and_B_hold}
    Consider a protocol $A \in \pidf$ that self-stabilizes to a strongly-connected topology $T$ and that fulfills the \mdl property.
    Then $A$ can be transformed into another protocol $B \in \pisf$ such that $B$ fulfills the stable \mdl property, $B$ self-stabilizes to the same topology, and in every computation of $B$ there exists a computation suffix in which Invariants~\ref{Inv:DelegateReq} and~\ref{Inv:DelegateAck} hold.
\end{theorem}
\begin{proof}
Consider the following transformation of $A$:
    \begin{enumerate}
     \item Every action $a$ in $A$ is transformed into an action $b$ such that every use of the Delegation primitive is replaced by the Safe-Delegation  primitive.
     That is: whenever in an action $a$ the current node $u$ removes a reference of another node $w$ from its memory and sends its to another node $v$ in a message $m$, action $b$ makes the edge $(u,w)$  temporary  and additionally sends out a \delegateREQ{u,v,u.eseq[v]} message to $v$.
For the whole transformation we will assume that $b$ always \emph{marks} its new messages and temporary edges for identification proposes, i.e., every temporary edge $e$ and every message $m$ of one of the newly created message types in $B$ carries an indication that the edge/message belongs to the original action $a$ and additionally carries the original message $m'$ (referenced by $e.origAction$/$m.origAction$ and $e.origMessage$/$e.origMessage$, respectively).
     Furthermore, any computation that $a$ performs on explicit edges, is performed on stable edges in $b$.
  \item  A node that receives a \delegateREQ{} message $m$ acts according to the Safe-Delegation primitive (where $m.origAction$ and $m.origMessage$ are also attached to the \delegateACK{} message sent) and executes the action that $m$ would trigger.
     \item After receiving a \delegateACK{w} message $m$, a node $u$ sends an \impldelegate{w} message to the stable neighbor $v$ that the edge $(u,w)$ would have been delegated to in $m.origAction$ (when its neighborhood is the stable neighborhood of $u$).
     \item If a node $u$ receives an \impldelegate{w} message, and $u$ already stores an edge $(u,w)$, $u$ does nothing. 
     Else if $u$ receives an \impldelegate{w} message $m$, and the action that $m.origMessage$ triggers in $A$ would cause $u$ to create the explicit edge $(u,w)$, then $u$ stores the reference of $w$.
     In any other case, $u$ executes the action that would be triggered by $m.origMessage$ in $A$ with the difference that it does not perform any changes to the state of $u$ (aside from removing the message from $u.Ch$), i.e., it only sends out messages.   
     \item In the \timeout action, a node $u$ sends a \delegateREQ{u,w} message  for each temporary edge $e=(u,w)$ to the stable neighbor $v$ to which that edge would be delegated to in $e.origAction$.
    \end{enumerate}
With a slight abuse of notation we say that a message $m$ sent by a node $u$ in protocol $A$ \emph{represents a Delegation} if $m$ is a message sent by $u$ in the execution of the Delegation primitive, i.e., in case $u$ delegates $w$ to $v$ (see Figure~\ref{fig:prim-dele}) $m$ is the message in $v.Ch$ holding the reference of $w$.

First, notice that $B \in \pisf$ follows directly from the construction of $B$.
Second, the fact that $B$ fulfills the stable \mdl property follows directly from the fact that $A$ fulfills the \mdl property.
Third, we prove the convergence property of $B$:

Consider an arbitrary but fixed computation $C_B$ of $B$ starting from a weakly connected graph.
Let $S$ be the first state of $C_B$ such that for all pairs of nodes $u, v$ one of the following two statements holds: (i) there was at least one state before $S$ in which an (implicit or explicit) edge $(u,v)$ existed or (ii) in no state of the computation suffix starting in $S$ an edge $(u,v)$ will exist.
Informally speaking $S$ is the state in which every node has received every reference it will ever receive in $C_B$ at least once.
Since the set of nodes is finite, such a state has to exist.
The stable \mdl property implies that the graph $ENG_s$ formed by the stable edges will remain the same in every state $S' > S$.
Thus, we now show that the graph $ENG_s(S)$ formed by the stable edges in $S$ is exactly $T$.

Consider that starting from $S$, every \delegateACK{} message is delivered without delay and the only actions that are executed are those that are triggered by \delegateACK{} messages.
Note that the only occasion at which \delegateACK{} messages are created is when a node receives a \delegateREQ{} message.
Consequently, there is a state $S' > S$ that is reachable from $S$ in which there is no \delegateACK{} message.
Analogously, it easy to show that there is a state $S'' > S'$ that is reachable from $S'$ in which there is neither a \delegateACK{} nor a \impldelegate{} message.

We transform $S''$ in the following way into a state $S_A''$:
Every temporary edge $e = (u,w)$  is replaced by a message $m.origMessage$ in the channel of the node $v$ that $u$ would send $m.origMessage$ to in $m.origAction$.
Similarly, every \delegateREQ{u,w,eseq} message in $v.Ch$ in $S''$ is replaced by $m.origMessage$ in $v.Ch$.
Loosely speaking, we create $S_A''$ by retransforming the temporary edges and \delegateREQ{} messages back to their original counterparts in $A$.
Note that $ENG(S_A'')=ENG_s(S'')$ and $NG(S_A'') = NG(S'')$.
Since there are no \delegateACK{}, \delegateREQ{} and \impldelegate{} messages in $S_A''$, $S_A''$ is a valid initial state for $A$.
Now consider a computation $C_A$ of $A$ that starts in $S_A''$.
We say a state $S_A$ of $A$ and a state $S_B$ of $B$ are \emph{equivalent} or short $S_A \equiv S_B$, if and only if (i) the local state of the nodes in $S_A$ (except for the edges) is equal to the local state of the nodes in $S_B$ (except for the edges) and (ii) the set of explicit edges in $S_A$ is equal to the set of stable edges in $S_B$ and (iii) for every message $m$ of $A$ that was sent in an action $\bar{a}$ and represents a Delegation (i.e., it is currently in $v.Ch$ for some node $v$, was sent by a node $u$, and carries a reference of a node $w$), there is a temporary edge $e = (u,w)$ with $e.origMessage = m$ and $e.origAction = \bar{a}$ and at least one \delegateREQ{u,w} message $m'$ with $m'.origMessage = m$ and $m'.origAction = \bar{a}$ in $v.Ch$ and (iv) no \delegateACK{} and \impldelegate{} messages exist in $S_B$.
We show that for every state $S_A^* \ge S_A''$ in $C_A$ there is a state $S^*$ reachable from $S''$ by $B$ such that $S^* \equiv  S_A^*$.

The induction base is obvious due to the construction of $S_A''$ from $S''$ .
For the induction step, consider an arbitrary but fixed state $S_A > S_A''$ in $C_A$ and let $\bar{a}$ be the action of $A$ that is executed by some node $u$ resulting in $N_A$, i.e., the subsequent state of $C_A$ (since we assume atomic action execution, this action is well-defined).
Then by construction of $B$ and the induction hypothesis, let $S_B$ be the state reachable from $S''$ with $S_B \equiv S_A$.
First consider the case that $\bar{a}$ is a \timeout action of $u$.
Consequently, we can execute the corresponding \timeout of $B$ for $u$. 
Both actions behave similar with the exception that (i) the \timeout of $B$ sends a \delegateREQ{v,w} message $m$ with $m.origMessage = m'$ and $m.origAction = \bar{a}$ whenever $\bar{a}$ sends a Delegation $m$ and (ii) the \timeout of $B$ sends additional \delegateREQ{u,w} messages $m$ for all temporary edges $e$ with $m.origAction = e.origAction$ and $m.origMessage = e.origMessage$.
Let $N_B$ be the subsequent state of $S_B$ after executing the \timeout of $B$.
The first difference makes sure that for each message $m$ sent in $\bar{a}$ that represents a Delegate there is a corresponding temporary edge $e$ with $e.origMessage = m$ and $e.origAction = \bar{a}$ and a \delegateREQ{} message $m'$ in $N_B$ with $m'.origMessage = m$ and $m'.origAction = \bar{a}$.
The fact that $B$ sends out additional \delegateREQ{u,w} messages for all temporary edges can be ignored since by definition of $S_B$ there already exists at least one \delegateREQ{u,w} for each temporary edge $(u,w)$ and they carry the same properties.
Thus, $N_A \equiv N_B$ holds.
Next consider the case that $\bar{a}$ is a not a \timeout action but an action triggered by a message $\bar{m}$ in $u.Ch$.
In case $\bar{m}$ does not represent a Delegate, then similar arguments as for the timeout action are applicable (since the action of $B$ sends \delegateREQ{v,w} whenever $\bar{a}$ uses the delegation primitive).
Therefore, consider the case that $\bar{m}$ represents a Delegate. 
Starting from $S_B$ consider the following action execution of $B$: 1) $u$ receives a \delegateREQ{} message $m'$ and answers with \delegateACK{} message $m''$ with $m''.origMessage = m'.origMessage$ and $m''.origAction = m'.origAction$ to the node $w$ that sent out the \delegateREQ{} (and also executes the action that would have been triggered by $\bar{m}$). 2) Node $w$ receives the \delegateACK{} message $m''$, removes the temporary edge and sends the \impldelegate{} $m'''$ with $m'''.origMessage = m''.origMessage$ and $m'''.origAction = m'''.origAction$ back to $u$. 3) Node $u$ receives the \impldelegate{} and does nothing because the delegated edge already exists in $u$.
By performing this action execution until all \delegateREQ{} messages $m$ in $u.Ch$ with $m.origMessage = \bar{m}$ are gone, we reach a state $M_B$ that is reachable from $S_B$ with the proprerty that (i) no \delegateACK{} and \impldelegate{} messages exist and (ii) the temporary edge and one \delegateREQ{} corresponding to $\bar{m}$ are gone. 
Moreover, $B$ did everything $\bar{a}$ would have done (since the same action is triggered by the \delegateREQ{} message), again with the difference that everytime $\bar{a}$ uses the Delegation primitive, $B$ performs the already described changes.
As we already argued for \timeout, these changes adhere to our desired statement, such that $N_A \equiv M_B$.

To show that $ENG_s(S') = T$, assume that there is a state $S_A^* > S_A''$ in which $A$ adds an explicit edge to $ENG$ or removes an explicit edge from $ENG$.
Then there is also a state $S^*$ that is reachable from $S'$ such that $S^*  \equiv  S_A^*$ (i.e., the set of stable edges of $S^*$ is equal to the set of explicit edges of $S_A^*$).
However, this would contradict to the fact that the graph formed by the stable edges cannot change anymore after $S'$.
On the other hand, since $A$ self-stabilizes $T$ from any weakly-connected graph and none of the transformations destroys the connectivity, we have that $ENG_s(S') = T$.

Therefore, to prove convergence it remains to show that as soon as $ENG_s(S') = T$ holds in a computation of $B$, all temporary edges will eventually vanish.
Observe that in $S'$ and every later state, any \delegateREQ{u,w,SEQ} message initiated by a node $u$ is sent to the same stable neighbor $v$ (since the stable neighborhood of $u$ does not change any more).
By the \mdl property and the construction of $B$, for any temporary edge $(u,w)$ there is a fixed cycle-free path $P_{(u,w)} = (u, u_2, u_3, \dots, u_k)$ via stable edges that this edge takes in $A$ until it is fused with the stable edge $(u_k,w)$.
We define the potential $\Phi := \sum_{node\ w} max_{u: (u,w)\ is\ temporary}(|P_{(u,w)}|-1)$ (where $max_{u: (u,w)\ is\ temporary}(|P_{(u,w)}|-1)$ is $0$ if there is no such edge) and show that as long as $\Phi$ is not zero, it will decrease in finite time.

For an arbitrary node $w$, consider a temporary edge $(u,w)$ such that $|P_{(u,w)}| =: l$ is maximal.
According to the protocol, $u$ will send a \delegateREQ{u,w,SEQ} to $u_2$ that will be answered with a \delegateACK{w,SEQ} message causing $u$ to make $(u,w)$ implicit and forwarding it to $u_2$.
The same happens for all other temporary edges $(v,w)$ with $|P_{(v,w)}| = l$.
After the last such edge has been removed from $v$, $\Phi$ decreases.
Note that $\Phi$ cannot increase since new temporary edges $x,y$ with $|P_{x,y}| > 1$ can only be caused by existing temporary / stable edges.
Thus, $\Phi$ will eventually be zero, meaning that there will be no more temporary edges.
Since no \delegateREQ{} and \delegateACK{} message will be sent from this state on, Invariants~\ref{Inv:DelegateReq} and~\ref{Inv:DelegateAck} hold in this state.

For the closure property, note that once the topology has stabilized no more temporary edges will exist and no more delegations will occur, implying that $B$ behaves exactly as $A$.
Since $A$ fulfills the closure property (and thus maintains the desired topology once stabilized), $B$ also does.
\end{proof}

\section{The generic search protocol}\label{sec:generic_search_protocol}
In this section we describe a generic search protocol such that every protocol in $\pisf$ fulfilling the stable \mdl property satisfies monotonic searchability according to that search protocol.
We assume that when a node $u$ wants to search for a node with identifier $ID$, it performs an \initsearch{ID} action in which a \search{u,ID} message is created.
The search request is regarded as answered as soon as the \search{u,ID} message is either dropped, i.e., it \emph{fails}, or is received by the node $w$ with $id(w)=ID$, i.e. it \emph{succeeds}.

The principle idea of the \emph{generic search protocol} is the following: 
A node $u$ with a \search{u,ID} message does not directly forward this message through the network but buffers it.
Instead, $u$ initiates a probing algorithm whose goal is to either receive the reference of the node $w$ with $id(w)=ID$, or to get a negative response in case this node does not exist or cannot be reached yet.
In the former case, $u$ directly sends \search{u,ID} to $w$.
In the latter case, $u$ drops \search{u,ID}.
Whenever an additional \search{u,ID} message for the same identifier $ID$ is initiated at $u$ while a probing for $ID$ is still in progress, this message is combined with previous \search{u,ID} messages waiting at $u$.

For the probing, a node $u$ with a buffered \search{u,ID} message  periodically initiates a new \probe{} message in its \timeout action.
This \probe{} message contains four arguments:
First, a reference $source$ of the source of the \probe{} message., i.e., a reference of $u$.
Second, the identifier $destID$ of the node that is searched, i.e., $ID$.
Third, a set $Next$ that holds references of all neighbors of $u$ with a closer distance to $destID$ than $id(u)$.
Last, a sequence number $seq$ that is used to distinguish probe messages that belong to different probing processes from the same node and for the same target, i.e., $seq = u.seq[ID]$, where $u.seq[ID]$ is a value stored at $u$.
This is necessary because in each execution of the \timeout action, a new probe message is sent, although upon receival of the first response to such a message, the set of buffered search messages is sent out to the target or dropped completely.
Thus, future replies may arrive afterwards and $u$ has to know that these are outdated.
All in all, $u$ initiates a \probe{source,destID,Next,seq} message and sends this message to the node in $Next$ whose identifier has the maximum distance to $ID$ (i.e., it is the closest to $u$).

Any intermediate node $v$ that receives a \probe{source,destID,Next,seq} message first checks whether $id(v) = destID$.
If so, $v$ sends a reference of itself to $source$ via a \psuccess{destID,dest} message with $dest = v$.
Otherwise, $v$ removes itself from $Next$ and adds all its neighbors to $Next$ that have a closer distance to $destID$ than itself.
If $Next$ is empty after this step, $v$ responds to $source$ via a \pfail{destID,seq} message.
Otherwise, $v$ forwards the \probe{source,destID,Next,seq} message (with the already described changes performed to $Next$) to the node in $Next$ whose identifier has the maximum distance to $ID$.
If the initiator $u$ of a probe receives a \psuccess{destID,dest} or a \pfail{destID,seq} message, it first checks whether $seq \geq u.seq[destID]$, i.e., it checks whether the received message is a response to the current batch of search requests.
If it is from an earlier probe, $u$ simply drops the received message.
Otherwise, $u$ acts depending on the message it received:
In case of a \psuccess{destID,dest} message, $u$ sends out all (possibly combined) \search{u,destID} messages waiting at $u$ to $dest$ (thus stopping the probing).
In case of a \pfail{destID,seq} message, $u$ drops all \search{u,destID} messages waiting  at $u$ to $dest$ (thus also stopping the probing).
In both cases, $u$ additionally increases $u.seq[destID]$ such that probe messages that are still in the system at this point in time cannot have any effects on future requests.
The pseudocode of the generic search protocol and supplementary details can be found in Listing~\ref{algo:search}.

\begin{lstlisting}[mathescape=true,caption= Generic \textsc{Search} protocol,label=algo:search]
$\initsearch{destID}$
 create new message $m=\search{self,destID}$
 if($WaitingFor[destID]=\emptyset$)
   $WaitingFor[destID] \gets \{\}$
   $self.seq \gets self.seq + 1$
   $seq[destID] \gets self.seq$
 //Store the messages to $WaitingFor$ 
 $WaitingFor[destID] \gets WaitingFor[destID] \cup \{m\}$ 

$\probe{source,destID,Next,seq}$
 if($destID = id(self)$)
   if($Next \neq \emptyset$) //can only occur in initial states
     for all $u \in Next$
       send $\impldelegate{u}$ to $self$
   send $\psuccess{destID, self}$ to $source$
   send $\impldelegate{source}$ to $self$
 else //$destID \neq id(self)$
   $Next \gets Next\setminus \{self\} \cup \{ neighbors\ w : dist(id(w),destID) < dist(id(self),destID) \}$
   if($Next = \emptyset$)
     send $\pfail{destID, seq}$ to $source$
     send $\impldelegate{source}$ to $self$ //to maintain connectivity
   else //$Next \neq \emptyset$
     $u \gets argmax\{ dist(id(u),destID) : u \in Next\}$
     if($u$ is not a neighbor of $self$)
       send $\impldelegate{u}$ to $self$
     send $\forwardprobe{source,destID,Next,seq}$ to $u$
 
$\psuccess{destID,dest}$
   send all $m \in WaitingFor[destID]$ to $dest$
   $WaitingFor[destID] \gets \emptyset$
 send $\impldelegate{dest}$ to $self$
 
$\pfail{destID,seq}$
 if($seq \geq seq[destID]$)
   /* The message belongs to the current set 
    * of buffered search requests to $dest$. */
   $WaitingFor[destID] \gets \emptyset$
\end{lstlisting}

Using the protocol as specified above could cause a high dilation because each probe message in each step is always sent to the node with the highest distance to the target in $Next$, even if a shorter path is possible.
Luckily, if there exists a fast routing protocol for the stabilized target topology (i.e., $o(n)$ hops in the worst case), it is possible to speed up search messages in legitimate states (and possibly even earlier).
To achieve this, whenever a node $u$ executes \initsearch{ID}, $u$ also creates a \fastprobe{source,destID} message with $source = u$ and $destID = ID$.
This message is routed according to the fast routing protocol.
If the fast routing protocol is not able to deliver the message to a node with id $destID$, the message is simply dropped.
If, however, the message is successfully delivered to a node $v$ with $ID(v) = destID$, this node responds to $source$ with a \psuccess{destID,v} message after which $u$ acts as specified before.

As the generic search protocol cannot guarantee to function properly under the presence of corrupt messages, we define the following additional invariants that are maintained during the execution of the generic search protocol (that did not start with corrupt messages):
\begin{enumerate} \setcounter{enumi}{2}
 \item If there is a \probe{source,destID,Next,seq} message in $u.Ch$, then
    \begin{enumerate}
      \item $u \in Next$ and $\forall w \in Next \setminus \{u\}: dist(id(w),destID) \leq dist(id(u),destID)$,
      \item $R(Next,ID) \subseteq R(source,ID)$, and
      \item if $v$ exists such that $id(v) = destID$ and $v \notin R(Next,destID)$, then for every admissible state with $source.seq[destID] < seq$, $v \notin R(source,destID)$.
    \end{enumerate}
     If there is a \fastprobe{source, destID} message in $u.Ch$, then 
     \begin{enumerate} \setcounter{enumii}{3}
	\item $u \in R(source, destID)$.
     \end{enumerate}
 \item If there is a \psuccess{destID, dest} message in $u.Ch$, then $id(dest) = destID$ and $dest \in R(u,destID)$.
 \item If there is a \pfail{destID, seq} message in $u.Ch$, then if $v$ exists such that $id(v) = destID$, then for every admissible state with $u.seq[destID] < seq$, $v \notin R(u,destID)$.
 \item If there is a \search{v, destID} message in $u.Ch$, then $id(u) = destID$ and $u \in R(v,destID)$.
\end{enumerate}

We say a protocol for the self-stabilization of a topology is \emph{monotonic-searchability-sufficient} (\emph{\mss}) if (i) all interactions between processes can be decomposed into the primitives in $\mathcal{ISF}$, (ii) it fulfills the stable \mdl property, (iii) it uses the generic search protocol for searching, (iv) no \probe{}, \psuccess{}, \pfail{}, or \search{} message is sent at any other occasion than the ones specified in the generic search protocol, and (v) in every computation of the protocol there is a state in which the first two invariants hold.

Note that Theorem~\ref{thm:conventional_protocols_transformable_such_that_invarA_and_B_hold} implies the following:
\begin{corollary}\label{cor:transformable}
 Any conventional protocol $A \in \pidf$ that self-stabilizes to a strongly-connected topology $T$ and that fulfills the \mdl property can be transformed into an \mss protocol that stabilizes the same topology.
\end{corollary}
For an \mss protocol, we define a state as admissible if all six invariants hold.
In the rest of this section, we prove the following theorem:
\begin{theorem}\label{thm:mss:satisfies:ntms}
 Every \mss protocol satisfies monotonic searchability according to Invariant~1-6.
\end{theorem}
On a side note, the following result follows directly from the description of the generic search protocol:
\begin{corollary}
 Every \mss protocol $P$ that stabilizes to a topology $T$ and in which the generic search protocol uses a routing strategy with a worst-case routing time of $O(T(n))$ for the fast search as described in the protocol, then $P$ answers successful search requests in legitimate states in time $O(T(n))$.
\end{corollary}

In order to prove Theorem~\ref{thm:mss:satisfies:ntms}, we need the following result, which is a corollary from Lemma~\ref{lem:once_A_and_B_hold_then_in_every_subsequent_step_as_well} and (v) of the \mdl property:
\begin{corollary}\label{cor:once_in_R_then_always}
 In any computation of an \mss protocol that starts from a state in which the first two invariants hold, if $v \in R(u,ID)$ in state $S$ then $v \in R(u,ID)$ in every state $S' \geq S$.
\end{corollary}

We first prove the following lemma:
\begin{lemma}\label{lem:admissible_suffix}
    In any computation of an \mss protocol, if there is an admissible state $S$, then all subsequent states will be admissible as well.
    Furthermore, in every computation of an \mss protocol, an admissible state exists.
\end{lemma}

The following sequence of lemmata constitutes the proof of Lemma~\ref{lem:admissible_suffix}.
We say a message $m$ \emph{causes} a message $m'$ if the message $m'$ was sent in an action handling the receipt of $m$ or a message caused by $m$.
\begin{lemma}\label{lem:first_three_hold}
    In any computation of an \mss protocol, if the first three invariants hold in a state $S$, they will hold in every state $S' \geq S$.
    Furthermore, in every computation of an \mss protocol, there exists a state in which the first three invariants hold.
\end{lemma}
\begin{proof}
    Consider a state $S''$ in which the first two invariants hold.    
    We show that the following six statements hold in $S''$ and every subsequent state:
 \begin{enumerate}
 \item Every \probe{} message sent in the \timeout action conforms to Invariant~3a)-c).
 \item Every \fastprobe{} message sent in the \initsearch{} action conforms to Invariant~3d).
 \item Every \probe{} message sent in the action executed on receipt of a \probe{} message conforms to Invariant~3a).
 \item Every \probe{} message sent in the action executed on receipt of a \probe{} message that conforms to Invariant~3b) and Invariant~3c) also conforms to Invariant~3b) and Invariant~3c).
 \item Every \fastprobe{} message sent in the action executed on receipt of a \fastprobe{} message that conforms to Invariant~3d) also conforms to Invariant~3d).
 \item Every \probe{} message violating the third invariant can cause at most a finite number of \probe{} message that violate Invariant~3b) and~3c).
\end{enumerate}
 
 For the first statement, observe that by construction the protocol ensures that every \probe{source, destID, Next, seq} sent in \timeout conforms to Invariant~3a) and~3b).
 Furthermore, $source.seq[destID]$ is monotonically increasing for every node $u$ and every $destID$.
 By Corollary~\ref{cor:once_in_R_then_always}, for any node $v$, $v \notin R(Next,destID)$ implies that $v \notin R(source,destID)$ in every previous state and, in particular, every state with $source.seq[destID] < seq$. 
 For similar reasons, the second statement also holds.
 
 The third statement follows directly from the protocol.
 
 For the fourth statement, consider an arbitrary \probe{source, destID, Next, seq} message received by a node $w$ that conforms to Invariant~3b) and Invariant~3c) and that causes the sending of a new \probe{source, destID, Next', seq} message(for some set $Next'$).
 According to the protocol, $R(Next',destID) \cup \{w\} = R(Next,destID)$, thus $R(Next',destID) \subseteq R(Next,destID) \subseteq R(source,destID)$ (the last set relationship is due to Invariant~3b)) and Invariant~3b) holds for the new message.
 Now assume that $v$ exists such that $id(v) = destID$ and $v \notin R(Next', destID)$.
 If $w = v$, $w$ would not have sent the new message, thus we may assume $w \neq v$.
 This yields $v \notin R(Next,destID)$ (since $R(Next,destID) = R(Next',destID) \setminus \{w\}$).
 By Invariant~3c) of the message received at $w$, we have that Invariant~3c) also holds for the message sent from $w$.
 
 For the fifth statement, consider an arbitrary \fastprobe{source, destID} message received by a node $w$ that conforms to Invariant~3d) and that causes the sending of a new \fastprobe{source, destID} message to some node $v$.
 Since $v$ is a neighbor of $w$, $v \in R(w,destID)$.
 By Invariant~3d), $w \in R(source,destID)$, thus also $v \in R(source,destID)$.
 By Corollary~\ref{cor:once_in_R_then_always}, this can never be rendered untrue in any subsequent state. 
 
 For the sixth statement, note that every \probe{source, destID, Next, seq} message received at a node $v$ causes at most one new \probe{source, destID, Next', seq} message sent from $v$ with the property that $max_{u' \in Next'} dist(u', destID) \leq max_{u \in Next} dist(u, destID)$ (due to the protocol and Invariant~3a)).
 Since the number of nodes is finite, this implies the \probe{source, destID, Next, seq} message can cause only a finite number of \probe{} message, which finishes the proof of the sixth statement.

 For the first claim of the lemma, assume there is a state $S$ such that the first three invariants hold. 
 Then the first five statements and Lemma~\ref{lem:once_A_and_B_hold_then_in_every_subsequent_step_as_well} yield that they will hold in every state $S' \geq S$.
 
 For the second claim of the lemma, note that a state $S''$ as defined above exists in every computation of an \mss protocol (according to the definition of an \mss protocol).
 The six statements yield that messages initiated in \timeout and messages conforming to Invariant~3 can only cause new messages conforming to Invariant~3 and that all other messages will eventually be gone.
 Thus, there is a state $S' \geq S''$ such that all messages conform to Invariant~3.
\end{proof}

\begin{lemma}\label{lem:first_five_hold}
    In any computation of an \mss protocol, if the first five invariants hold in a state $S$, they will hold in every state $S' \geq S$.
    Furthermore, in every computation of an \mss protocol, there exists a state in which the first five invariants hold.
\end{lemma}
\begin{proof}
   Consider a state $S''$ in which the first three invariants hold.    
   By Lemma~\ref{lem:first_three_hold}, they will also hold in every subsequent state.
 We show that in $S''$ and every subsequent state, no receipt of a \probe{} message causes the sending of a \psuccess{} or \pfail{} message that violates Invariant~4, or Invariant~5, respectively.
 Observe that according to the protocol, no other action causes a node to send a \psuccess{} or \pfail{} message.

 First we show that no receipt of a \probe{} or a \fastprobe{} message can cause the sending of a \psuccess{} message  violating Invariant~4.
 To be more specific, consider a node $v$ that receives a \probe{source, destID, Next, seq} message and this causes $v$ to send a \psuccess{destID',dest'} message to a node $u$.
 In this case, $id(v) = destID$ must hold.
 Furthermore, according to the protocol, $destID' = destID$, and $dest' = v$.
 By Invariant~3a), when $v$ receives the \probe{} message, $v \in Next$ and by Invariant~3b), $R(Next,destID) \subseteq R(source,destID)$, i.e. $v \in R(source,destID)$.
 Thus, the \psuccess{} message sent by $v$ conforms to Invariant~4.
 For the other case, assume that a node $v$ receives a \fastprobe{source, destID} message and this causes $v$ to send a \psuccess{destID',dest'} to a node $u$.
 In this case, also $id(v) = destID$ must hold.
 Furthermore, according to the protocol, $destID' = destID$, and $dest' = v$. 
 By Invariant~3d), when $v$ receives the \fastprobe{} message, $v \in R(source,destID)$.
 Thus, the \psuccess{} message sent by $v$ conforms to Invariant~4 in this case, too.
 
 Second we show that no receipt of a \probe{} message can cause the sending of a \pfail{} message that violates Invariant~5.
 Again, consider a node $v$ that receives a \probe{source, destID, Next, seq} message causing $v$ to send a \pfail{destID',seq'} to a node $u$.
 According to the protocol $destID' = destID$, $seq' = seq$ and $u = source$ must hold.
 If there is a node $w$ such that $id(w) = destID$, then by Invariant~3c) for the \probe{source, destID, Next, seq} message received by $v$, for every admissible state $u.seq[destID] < seq$, $w \notin R(u,destID)$.
 Thus, the \pfail{} message sent by $v$ conform to Invariant~5.
 
 For the first claim of the lemma, assume there is a state $S$ such that the first five invariants hold.
 By what we showed above, they will also hold in every state $S' \geq S$.
 
 For the second claim of the lemma, note that a state $S''$ as defined above exists in every computation of an \mss protocol by Lemma~\ref{lem:first_three_hold}.
 Since we showed that no new message can violate Invariant~4 or Invariant~5, the state $S' \geq S''$ in which all \psuccess{} or \pfail{} messages that violate Invariant~4 or Invariant~5 and that are in incoming message channels in state $S''$ have been received is a state in which the first five invariants hold.
\end{proof}

Finally, we can use these results to prove Lemma~\ref{lem:admissible_suffix}:
\begin{proof}
    Consider a state $S''$ in which the first five invariants hold.
    By Lemma~\ref{lem:first_five_hold}, they will also hold in every subsequent state.   
 We show that in $S''$ and every subsequent state, no \search{} message that violates Invariant~6 can be sent.
 
 According to the protocol, the only occasion at which a \search{v,destID} message is sent to a node $u$ is if node $v$ received a \psuccess{destID',dest} message with $destID' = destID$ and $u = dest$.
 Invariant~4 yields that $id(dest) = destID'$, implying $id(dest) = destID$, and $dest \in R(v,destID')$, implying $u \in R(v,destID)$.
 This implies that every new message conforms to Invariant~6, which completes the proof of the first claim of the lemma.

 For the second claim of the lemma, note that a state $S''$ as defined above exists in every computation of an \mss protocol by Lemma~\ref{lem:first_five_hold}.
 Consider the state $S \geq S''$ in which all \search{} messages violating Invariant~6 that are in incoming message channels in $S''$ have been received.
 $S$ is a state such that the first six invariants hold, i.e., $S$ is an admissible state.
\end{proof}

Next, we are going to prove the following lemma:
\begin{lemma}\label{lem:in_admissible_search_correct}
 If a node $u$ has a \search{u,destID} message waiting at $u$ in an admissible state, then this message will eventually be delivered or dropped.
 In the former case, it will be delivered to the node $w$ with $id(w) = destID$ (which exists in this case).
 In the latter case, either there is no node with id $destID$ or all earlier \search{u,destID} messages that have been waiting at $u$ in admissible states have been or will be dropped as well.
 \end{lemma}
We define the potential $\Psi(U,ID) := \sum_{u \in U} n^{dist(u,ID)}$, where $n$ is the number of processes and prove the following lemma:
\begin{lemma}\label{lem:in_admissible_distance_decreasing}
 Consider an admissible state, if a node $u$ receives a \probe{v,destID, Next, seq} message and $R(Next,destID) \setminus \{u\} \neq \emptyset$, then $u$ sends a \probe{v,destID, Next', seq} message to some other node, 
 and $\Psi(Next', destID) < \Psi(Next, destID)$.
\end{lemma}
\begin{proof}
 According to the protocol, when $u$ receives a \probe{v, destID, Next, seq} message, $u$ chooses $Next'$ as $Next \setminus \{u\}$ augmented by neighbors whose distance to $distID$ is smaller than $dist(u,distID)$.
 Thus, $\Psi(Next', destID) < \Psi(Next, destID)$.
\end{proof}

This enables us to prove:
\begin{lemma}\label{lem:in_admissible_probe_always_answered}
 If a node $v$ initiates a \probe{v, destID, Next, seq} message in an admissible state, then $v$ eventually receives either a \psuccess{destID, dest} message for some node $dest$ or a \pfail{destID, seq} message.
\end{lemma}
\begin{proof}
    Lemma~\ref{lem:in_admissible_distance_decreasing} yields that each \probe{v, destID, Next', seq} message eventually causes a \probe{v, destID, Next'', seq} message with $\Psi(Next'', destID) < \Psi(Next', destID)$.
    Moreover, by definition, $\Psi(Next', destID)$ can not increase for a \probe{v,destID, Next', seq} while it is in an incoming channel of a node.
    By induction and because $\Psi$ is bounded below, this yields that there must a node $w$ that does not send a new \probe{} message upon receipt of a \probe{v, destID, Next', seq} message caused by the original \probe{v, destID, Next, seq} message sent by $v$.
    According to the protocol, if $w$ does not send such a message, it sends either a \psuccess{destID, w} message or a \pfail{destID, seq} message to $v$.
\end{proof}

Using this, we prove Lemma~\ref{lem:in_admissible_search_correct}:
\begin{proof}
    Assume there is an admissible state $S$ and a node $u$ that has a \search{u,destID} message $m$ waiting at $u$ in $S$.
     Node $u$ initiates a \probe{u, destID, Next, seq} message every time it executes \timeout.
     According to Lemma~\ref{lem:in_admissible_probe_always_answered}, $u$ will eventually receive a \psuccess{destID, dest} or a \pfail{destID, seq} message.
     Note that $u$ forwards or drops $m$ upon the first receipt of such message after $S$.
     
     First, consider the case that the first such message that $u$  receives after state $S$  is a \psuccess{destID, dest} message.
     Invariant~4 (which holds due to Lemma~\ref{lem:admissible_suffix}) yields $id(dest) = destID$ and the $m$ message will be sent to $dest$, according to the protocol.
     
     Second, consider that the first such message that $u$ receives is a \pfail{destID, seq} message.
     According to Invariant~5, either no node $v$ with $id(v) = destID$ exists(in which case we are finished) or for every admissible state with $u.seq[destID] < seq, v \notin R(u,destID)$.
     Now consider an arbitrary earlier \search{u,destID} message $m'$ that has been waiting at $u$.
     If $m'$ is still waiting at $u$ in state $S$, then $m'$ will be dropped together with $m$ when $u$ receives the \pfail{destID, seq} message.
     Otherwise, assume for contradiction that in an admissible state $S'$, $m'$ was sent to a node $dest$ with $id(dest) = destID$.
     According to the protocol, this requires that $u$ received a \psuccess{destID,dest}.
     By Invariant~4, $dest \in R(u,destID)$ held in this state.
     Additionally, $u$ increased $u.seq[destID]$ upon receipt of that message.
     Since the sequence numbers are monotonically increasing, $S'$ is a state with $u.seq[destID] < seq$.
     Thus, Invariant~5 of the \pfail{destID, seq} message implies $dest \notin R(u,destID)$ in state $S'$, yielding a contradiction.
     This finishes the proof.        
\end{proof}

Last, we prove the following lemma:
\begin{lemma}\label{lem:u_in_R_then_delivered}
 If there is a node $v \in R(u,id(v))$ for a node $u$ in an admissible state $S$, then there will be a state $S' \geq S$ such that all \search{u,id(v)} messages initiated in $S'$ and all subsequent states will be delivered to $v$.
\end{lemma}
\begin{proof}
 Assume $v \in R(u,id(v))$ in $S$.
 Note that all \search{u,id(v)} messages waiting at $u$ are either delivered or dropped in some state $S'$ that is subsequent to $S$, according to Lemma~\ref{lem:in_admissible_search_correct}.
 Furthermore, in $S'$, $u.seq[id(v)]$ will have increased by one (compared to $S$).
 Now consider an arbitrary \search{u,id(v)} message $m$ initiated in $S'$ or any subsequent state.
 By Lemma~\ref{lem:in_admissible_search_correct} this will be delivered or dropped.
 Assume for contradiction it is dropped.
 This must have happened due to a \pfail{destID,seq} message (with $seq > u.seq[id(v)]$ in $S$).
 Thus, Invariant~5 for this message yields a contradiction to $v \in R(u,id(v))$ in $S$.
 All in all, we have that $m$ is delivered correctly.
\end{proof}

Using Lemma~\ref{lem:admissible_suffix}, Lemma~\ref{lem:in_admissible_search_correct} and Lemma~\ref{lem:u_in_R_then_delivered}, we can prove Theorem~\ref{thm:mss:satisfies:ntms}:
\begin{proof}
 Note that Lemma~\ref{lem:admissible_suffix} yields that in every computation of an \mss protocol that starts from an admissible state, there is an admissible state, and that in every computation of an \mss protocol there is an admissible state.
 Thus, it remains to prove that the protocol satisfies monotonic searchability in computations that start from admissible states.
 
 Consider an arbitrary computation that starts from an admissible state and assume there are two \search{u,destID} messages $m$ and $m'$ initiated in a node $u$ such that $m$ has been initiated before $m'$ and $m$ successfully reaches its target, but $m'$ is dropped.
 This is a contradiction to Lemma~\ref{lem:in_admissible_search_correct}, 
 Thus $m'$ must reach its target as well.
 
 Last, note that by definition, in legitimate states, $E(u,v) \Leftrightarrow R(u, id(v))$ for each pair of nodes $u,v$ and $ENG$ forms the desired topology.
 That is, for every two nodes $u$ and $v$ such that there is path from $u$ to $v$ in the target topology, $v \in R(u, id(v))$.
 Thus, Lemma~\ref{lem:u_in_R_then_delivered} yields that the monotonic searchability is non-trivial, which completes the proof of Theorem~\ref{thm:mss:satisfies:ntms}.
\end{proof}

\section{Examples}\label{sec:examples}
In this section, we give some examples of topologies for which a self-stabilizing protocol that satisfies non-trivial monotonic searchability according to the generic search protocol can be obtained by combining existing results in the literature with the results presented in this work.
\footnote{Note that some of the results we cite have been obtained for a synchronous model. However, the results easily transfer to the asynchronous model.}
By Corollary~\ref{cor:transformable} and Theorem~\ref{thm:mss:satisfies:ntms}, all that needs to be shown is that the target topology is strongly connected, that $E(u,v) \Leftrightarrow R(u,id(v))$ holds and that the protocols we want to adapt fulfill the \mdl property.

\subsection{The linear list}
The line topology is the first topology for which a protocol satisfying non-trivial monotonic searchability was shown \cite{SetzerSSS14}.
Now, using the universal approach described in this paper and applying it to the self-stabilizing list protocol presented by Nor, Nesterenko, and Scheideler (called $l$-Corona in \cite{Nor13Corona}), the non-trivial monotonic searchability follows as a corollary.

First of all, it is obvious that the linear list is strongly connected and that it holds $E(u,v) \Leftrightarrow R(u,id(v))$.
Second, note that according to the protocol in \cite{Nor13Corona}, whenever each node $u$ receives a new reference $v$ with $id(u) < id(v)$ ($id(v) < id(u)$), $u$ keeps the reference if it is closer to its previous right (left) neighbor $w$, in which case it also introduces $w$ to $v$, or delegates the reference to its previous right (left) neighbor, otherwise.
In a legitimate state, every reference is simply delegated into one direction until it finally fuses with an already existing edge.
Thus, the protocol fulfills the \mdl property.

\subsection{The \skipp graph}
The \skipp graph was introduced by Jacob, Richa, Scheideler, Schmid, and T{\"{a}}ubig \cite{JRSST09} as a supergraph of the skip graph to overcome the issue that the consistency of the original skip graph \cite{Normalskip} is not locally checkable.
For a formal definition of the \skipp graph and the self-stabilizing protocol for it, we refer the reader to that paper.

Obviously, the \skipp is strongly connected.
Besides, $E(u,v) \Leftrightarrow R(u,id(v))$ is easy to see because the \skipp graph contains as a subgraph the linear list (on level $0$). 
Thus, all that remains to be shown is that the protocol presented in \cite{JRSST09} fulfills the \mdl property.
We will argue about every of the four sub-properties individually.
To see that (i) is fulfilled, note that the protocol has a notion of stable and temporary edges that is close to ours and that stable edges are never delegated away (since the only Rule that delegates away an edge is Rule~2).
That is, each node always keeps references of nodes that appear to belong to the final topology (from its local view).
By the definition of a node's neighborhood, having a neighbor $w$ that does not belong to the final topology cannot cause an edge $(u,v)$ to appear to not belong to the final topology for a node $u$.
Thus, (i) is not violated.
The argument for (ii) is similar: Since a node always keeps its ``closest'' neighbors, receiving a previously known reference again cannot cause a node to add the corresponding node to its neighborhood.
Thus, (ii) is also fulfilled.
Sub-property (iii) is easy to see and for (iv) and (v), observe that any implicit edge $(u,v)$ in a legitimate state is always delegated to the neighbor that shares the largest common prefix with $v.rs$.
By the definition of the \skipp graph, this leads to this edge finally reaching a neighbor of $v$ at which is is fused.
All in all, the protocol fulfills the \mdl property and monotonic searchability is possible using our approach for the \skipp graph.

\section{Conclusion and Outlook}\label{sec:conclusion}
In this work we further strengthened the notion of monotonic searchability introduced in \cite{SetzerSSS14} by presenting a universal approach for adapting conventional protocols for topological self-stabilization such that they satisfy monotonic searchability.
Even more, we carved out some design principles that protocols should adhere to in order to enable reliable searches even during the stabilization phase.

Although our results solve the problem of monotonic searchability for a wide range of topologies, there are certain aspects that have not been studied yet.
For example, we did not consider the additional cost of convergence (i.e., the amount of additional messages to be sent), nor the impact of our methods on the convergence time of the topology.
Additionally, while our generic search protocol enables us to search existing nodes in legitimate states with a low dilation, searching for a non-existing node can still cause a message to travel $\Omega(n)$ hops, even in a legitimate state.
Whether this is provably necessary or could be improved is still an open question.

Last, in~\cite{SetzerSSS14} a solution for the linear list was presented that maintains monotonic searchability even under join/leave dynamics in terms of the Finite Departure Problem.
For all other topologies (and also for our general approach), such a solution is still missing.

\bibliography{bibliography.bib}

\begin{thebibliography}{10}

\bibitem{Normalskip}
James Aspnes and Gauri Shah.
\newblock Skip graphs.
\newblock {\em {ACM} Trans. Algorithms}, 3(4), 2007.

\bibitem{AspnesW07}
James Aspnes and Yinghua Wu.
\newblock O(logn)-time overlay network construction from graphs with out-degree
  1.
\newblock In {\em Principles of Distributed Systems, 11th International
  Conference, {OPODIS} 2007, Guadeloupe, French West Indies, December 17-20,
  2007. Proceedings}, pages 286--300, 2007.

\bibitem{DBLP:journals/tcs/BernsGP13}
Andrew Berns, Sukumar Ghosh, and Sriram~V. Pemmaraju.
\newblock Building self-stabilizing overlay networks with the transitive
  closure framework.
\newblock {\em Theor. Comput. Sci.}, 512:2--14, 2013.

\bibitem{DBLP:journals/dc/BuiDPV07}
Alain Bui, Ajoy~Kumar Datta, Franck Petit, and Vincent Villain.
\newblock Snap-stabilization and {PIF} in tree networks.
\newblock {\em Distributed Computing}, 20(1):3--19, 2007.

\bibitem{DBLP:journals/jpdc/DelaetDNT10}
Sylvie Dela{\"{e}}t, St{\'{e}}phane Devismes, Mikhail Nesterenko, and
  S{\'{e}}bastien Tixeuil.
\newblock Snap-stabilization in message-passing systems.
\newblock {\em J. Parallel Distrib. Comput.}, 70(12):1220--1230, 2010.

\bibitem{Dijkstra74}
Edsger~W. Dijkstra.
\newblock Self-stabilizing systems in spite of distributed control.
\newblock {\em Commun. {ACM}}, 17(11):643--644, 1974.

\bibitem{DBLP:journals/cjtcs/DolevH97}
Shlomi Dolev and Ted Herman.
\newblock Superstabilizing protocols for dynamic distributed systems.
\newblock {\em Chicago J. Theor. Comput. Sci.}, 1997, 1997.

\bibitem{DolevK08}
Shlomi Dolev and Ronen~I. Kat.
\newblock Hypertree for self-stabilizing peer-to-peer systems.
\newblock {\em Distributed Computing}, 20(5):375--388, 2008.

\bibitem{DolevT2013}
Shlomi Dolev and Nir Tzachar.
\newblock Spanders: Distributed spanning expanders.
\newblock {\em Sci. Comput. Program.}, 78(5):544--555, 2013.

\bibitem{self-stabilizing-list2}
Dominik Gall, Riko Jacob, Andr{\'{e}}a~W. Richa, Christian Scheideler, Stefan
  Schmid, and Hanjo T{\"{a}}ubig.
\newblock A note on the parallel runtime of self-stabilizing graph
  linearization.
\newblock {\em Theory Comput. Syst.}, 55(1):110--135, 2014.

\bibitem{JRSST09}
Riko Jacob, Andr{\'{e}}a~W. Richa, Christian Scheideler, Stefan Schmid, and
  Hanjo T{\"{a}}ubig.
\newblock Skip\({}^{\mbox{+}}\): {A} self-stabilizing skip graph.
\newblock {\em J. {ACM}}, 61(6):36:1--36:26, 2014.

\bibitem{JacobRSS2012}
Riko Jacob, Stephan Ritscher, Christian Scheideler, and Stefan Schmid.
\newblock Towards higher-dimensional topological self-stabilization: A
  distributed algorithm for delaunay graphs.
\newblock {\em Theor. Comput. Sci.}, 457:137--148, 2012.

\bibitem{DBLP:conf/europar/JohnenM10}
Colette Johnen and Fouzi Mekhaldi.
\newblock Robust self-stabilizing construction of bounded size weight-based
  clusters.
\newblock In {\em Euro-Par 2010 - Parallel Processing, 16th International
  Euro-Par Conference, Ischia, Italy, August 31 - September 3, 2010,
  Proceedings, Part {I}}, pages 535--546, 2010.

\bibitem{DBLP:conf/ipps/KakugawaM06}
Hirotsugu Kakugawa and Toshimitsu Masuzawa.
\newblock A self-stabilizing minimal dominating set algorithm with safe
  convergence.
\newblock In {\em 20th International Parallel and Distributed Processing
  Symposium {(IPDPS} 2006), Proceedings, 25-29 April 2006, Rhodes Island,
  Greece}, 2006.

\bibitem{KniesburgesKS12}
Sebastian Kniesburges, Andreas Koutsopoulos, and Christian Scheideler.
\newblock A self-stabilization process for small-world networks.
\newblock In {\em 26th {IEEE} International Parallel and Distributed Processing
  Symposium, {IPDPS} 2012, Shanghai, China, May 21-25, 2012}, pages 1261--1271,
  2012.

\bibitem{KoutsopoulosSS15}
Andreas Koutsopoulos, Christian Scheideler, and Thim Strothmann.
\newblock Towards a universal approach for the finite departure problem in
  overlay networks.
\newblock In {\em Stabilization, Safety, and Security of Distributed Systems -
  17th International Symposium, {SSS} 2015, Edmonton, AB, Canada, August 18-21,
  2015, Proceedings}, pages 201--216, 2015.

\bibitem{corona}
Rizal~Mohd Nor, Mikhail Nesterenko, and Christian Scheideler.
\newblock Corona: {A} stabilizing deterministic message-passing skip list.
\newblock {\em Theor. Comput. Sci.}, 512:119--129, 2013.

\bibitem{Nor13Corona}
Rizal~Mohd Nor, Mikhail Nesterenko, and Christian Scheideler.
\newblock Corona: {A} stabilizing deterministic message-passing skip list.
\newblock {\em Theor. Comput. Sci.}, 512:119--129, 2013.

\bibitem{self-stabilizing-list}
Melih Onus, Andr{\'{e}}a~W. Richa, and Christian Scheideler.
\newblock Linearization: Locally self-stabilizing sorting in graphs.
\newblock In {\em Proceedings of the Nine Workshop on Algorithm Engineering and
  Experiments, {ALENEX} 2007, New Orleans, Louisiana, USA, January 6, 2007},
  2007.

\bibitem{SetzerSSS14}
Christian Scheideler, Alexander Setzer, and Thim Strothmann.
\newblock Towards establishing monotonic searchability in self-stabilizing data
  structures.
\newblock In {\em Principles of Distributed Systems - 19th International
  Conference, {OPODIS} 2015, Proceedings}, 2015.

\bibitem{ShakerR05}
Ayman Shaker and Douglas~S. Reeves.
\newblock Self-stabilizing structured ring topology {P2P} systems.
\newblock In {\em Fifth {IEEE} International Conference on Peer-to-Peer
  Computing {(P2P} 2005), 31 August - 2 September 2005, Konstanz, Germany},
  pages 39--46, 2005.

\bibitem{YamauchiT10}
Yukiko Yamauchi and S{\'{e}}bastien Tixeuil.
\newblock Monotonic stabilization.
\newblock In {\em Principles of Distributed Systems - 14th International
  Conference, {OPODIS} 2010, Tozeur, Tunisia, December 14-17, 2010.
  Proceedings}, pages 475--490, 2010.

\end{thebibliography}

\end{document}